\RequirePackage{fix-cm}
\documentclass[smallcondensed]{svjour3}     
\smartqed  
\usepackage{graphicx}
\usepackage{mathtools}
\usepackage{enumitem}
\usepackage{comment}
\usepackage{amsmath}

\usepackage{amsthm}
\usepackage{amsfonts}
\usepackage{tensor}
\usepackage{bbm}
\usepackage{hyperref} 
\usepackage{breqn}
\hypersetup{colorlinks=true, linkcolor=blue, citecolor=blue}

\newcommand{\subparagraph}{}
\usepackage{titlesec}
\titlelabel{\thetitle.\,\,}
  
\def\equationautorefname~#1\null{\fontshape{n}\textup{\textcolor{black}{Eq. (}#1\textcolor{black}{)}}\null}
\def\theoremautorefname~#1\null{\fontshape{n}\textup{\textcolor{black}{Theorem }#1}\null}
\def\definitionautorefname~#1\null{\fontshape{n}\textup{\textcolor{black}{Definition} #1}\null}
\def\sectionautorefname~#1\null{\fontshape{n}\textup{\textcolor{black}{Section} #1}\null}
\def\subsectionautorefname~#1\null{\fontshape{n}\textup{\textcolor{black}{Section} #1}\null}
\def\lemmaautorefname~#1\null{\fontshape{n}\textup{\textcolor{black}{Lemma} #1}\null}
\let\oldcite\cite
\renewcommand{\cite}[1]{\textup{\oldcite{#1}}}

\makeatletter

\newcommand\Autoref[1]{\@first@ref#1,@}
\def\@throw@dot#1.#2@{#1}
\def\@set@refname#1{
    \edef\@tmp{\getrefbykeydefault{#1}{anchor}{}}%
    \xdef\@tmp{\expandafter\@throw@dot\@tmp.@}%
    \ltx@IfUndefined{\@tmp autorefnameplural}%
         {\def\@refname{\@nameuse{\@tmp autorefname}s}}%
         {\def\@refname{\@nameuse{\@tmp autorefnameplural}}}%
}
\def\@first@ref#1,#2{%
  \ifx#2@\autoref{#1}\let\@nextref\@gobble
  \else%
    \@set@refname{#1}
    \@refname~\ref{#1}
    \let\@nextref\@next@ref
  \fi%
  \@nextref#2%
}
\def\@next@ref#1,#2{%
   \ifx#2@ and~\ref{#1}\let\@nextref\@gobble
   \else, \ref{#1}
   \fi%
   \@nextref#2%
}

\makeatother

\newcommand{\ml}{\mathcal{M}_{\lambda}}
\newcommand{\lie}{\mathcal{L}}
\newcommand{\rt}{\overline{T}}
\newcommand{\rtwo}{{}^{(2)}\!R}
\newcommand{\rthree}{{}^{(3)}\!R}
\newcommand{\ko}{\ensuremath{K_{\Lambda_1}}}
\newcommand{\kt}{\ensuremath{K_{\Lambda_2}}}
\newcommand{\lo}{\ensuremath{\Lambda_1}}
\newcommand{\lt}{\ensuremath{\Lambda_2}}

\newcommand{\intmv}{\ensuremath{\int_{M_{\lambda}} dV_{\lambda} \,}}
\newcommand{\intmz}{\ensuremath{\int_{0}^{L(\lambda)} dz \,}}
\newcommand{\parl}{\ensuremath{\partial_{\lambda}}}

\newcommand{\ep}{\ensuremath{\mathcal{E}}}

\theoremstyle{definition}

\titlerunning{Asymptotics of 3+1D Cosmologies with potential $0 < \Lambda_1 \leq U \leq \Lambda_2$}

\AtBeginDocument{%
  \paperwidth=\dimexpr
    1in + \oddsidemargin
    + \textwidth
    + 1in + \oddsidemargin
  \relax
  \paperheight=\dimexpr
    1in + \topmargin
    + \headheight + \headsep
    + \textheight
    + 1in + \topmargin
  \relax
  \usepackage[pass]{geometry}\relax
}

\begin{document}

\title{ On the asymptotics of 3+1D cosmologies with bounded scalar potential and isometry group forming 2-dimensional orbits}


\author{Jinhui Wang         \and
       Leonardo Senatore 
}


\institute{Jinhui Wang \at
             Department of Physics, Stanford University \\
              \email{wangjh97@stanford.edu}           
           \and
           Leonardo Senatore \at
          Institute for Theoretical Physics, ETH Zurich  \\
            \email{lsenatore@phys.ethz.edu}  
}

\date{Received: date / Accepted: date}

\maketitle

\begin{abstract}
We study the onset of inflation in 3+1 dimensional cosmologies with an inflationary potential $U$ satisfying $0 < \Lambda_1 \leq U \leq \Lambda_2$, matter satisfying the dominant and strong energy conditions, and with spatial slices that can be foliated by 2-dimensional surfaces that are orbits under an isometry group. Assuming an initial Cauchy slice with positive mean curvature everywhere, we show, via mean curvature flow, that there exists a family of spatial slices parameterized by $\lambda$, whose volume grows between the flat slicings in de Sitter spaces with cosmological constants $\Lambda_1$ and $\Lambda_2$. In particular, inflationary expansion indeed occurs in this setting with inhomogeneous initial conditions. Finally, we apply this ``inflationary time coordinate'' $\lambda$ to study asymptotics of the variation in the metric, the average stress-energy tensor, and the dynamics of an inflaton field on a spatial slice.

\end{abstract}

\newpage

\section{Introduction}
Inflation has been a highly successful theory in resolving the horizon, flatness, and magnetic monopole problems in cosmology \cite{guth1981,Linde:1981mu,Albrecht:1982wi,Linde:1983gd}. Furthermore, it has shown remarkable agreement with observational data from the Cosmic Microwave Background (see for example~\cite{Aghanim:2018eyx,Akrami:2018odb,Akrami:2019izv}) and the Large-Scale Structure of the universe (see for example~\cite{DAmico:2019fhj,Ivanov:2019pdj,Colas:2019ret,Abbott:2017wau,Abbott:2018wzc}). Despite these triumphs, the standard analysis of inflation suffers from an unrealistic assumption of homogeneous initial conditions. Concretely, one assumes an initial region of space, on the order of the Hubble length during inflation, that is homogeneously filled with the inflationary scalar field (inflaton) at the top of its potential. Such fine-tuned initial conditions seem improbable, for which this issue is dubbed the ``initial patch problem'' (see for example~\cite{Ijjas:2015hcc}). More fundamentally, spatial homogeneity is widely thought to be a consequence of inflation, and not its cause. Hence, whether inflation can begin under inhomogeneous initial conditions has been an important question under debate.

Recently, significant progress on both numerical and analytic fronts have been made in this direction. On the numerical side, \cite{East:2015ggf,Clough:2016ymm,Clough:2017efm} have found that inflation always begins from a large set of inhomogeneous initial conditions, enabled by developments in numerical relativity that can handle singularities and horizons \cite{Pretorius:2005gq}. {These numerical experiments seem to go a long way in showing that inflation does not generically need an initial homogenous patch.} On the analytic end, \cite{Kleban:2016sqm} has attained partial results in the general case of a positive cosmological constant without extra symmetries, via a combination of Mean Curvature Flow techniques (see for example~\cite{gerhardtbook}) and the Thurston Geometrization Classification (see \cite{besse1987einstein} Theorem 4.35 and \cite{thurston1997three,10.2307/2152760}), under the assumptions of matter satisfying the weak energy condition and all space-time singularities being of the so-called crushing kind. Concretely, \cite{Kleban:2016sqm} showed that for almost all topologies of the spatial slices of a cosmology, the volume of these slices is increasing and that there is always an open neighborhood that expands at least as fast as the flat slicing of de Sitter space. This is suggestive that the volume will expand to infinity, matter will dilute away, and the universe will resemble de Sitter space in arbitrarily large regions of spacetime. These statements were proven in 2+1 dimensions (with the additional assumption that matter satisfies the strong and the dominant energy conditions in \cite{Creminelli:2019pdh} and stronger assumptions on the matter content and initial conditions in \cite{Barrow:2006cw}). In 3+1 dimensions, \cite{creminelli2020sitter} proved these results plus the pointwise convergence of the metric to that of de Sitter space and `physical equivalence' to de Sitter, under the additional assumption that there exists an isometry group forming 2-dimensional orbits on the spatial slices of the cosmology. In particular, the latter convergence is a concrete, affirmative instance of the de Sitter no-hair conjecture (see for example \cite{Gibbons:1977mu,Hawking:1981fz,Wald:1983ky,Kleban:2016sqm}) which states that initially sufficiently expanding cosmologies should asymptote to de Sitter space in the presence of a positive cosmological constant.

Though the above analytic results are already suggestive that the onset of inflation doesn't necessitate an initial homogeneous patch, all of them assumed a positive cosmological constant as a model of the underlying inflationary potential. In this paper, we focus on 3+1 dimensions with the isometry assumptions of \cite{creminelli2020sitter} and an inflationary potential $U$ satisfying $0 < \lo \leq U \leq \lt$, which supersedes the previously-supposed cosmological constant (which can be seen as the limiting case $\lo = \lt = \Lambda$). We show that if $\frac{\lt}{\lo} < \frac{3}{2}$, there is a family of spacelike hypersurfaces $\{ M_{\lambda} \}$, parameterized by a flow time $\lambda$, that reaches arbitrarily late times (as defined by the Cauchy time function of the cosmology) and whose volumes grow between those de Sitter spaces with cosmological constants $\lo$ and $\lt$ in the flat slicing. In other words, inflationary expansion indeed occurs under our more general assumptions, with $\lambda$ as a new time coordinate.

We also present partial results for the variation of the metric on $\ml$, the dilution of matter, and the asymptotic dynamics of an inflaton. These bounds are suggestive (but do not prove) that the metric approaches a spatially homogeneous one on the slices $M_{\lambda}$ for late enough $\lambda$. All-in-all, our results govern the asymptotics of a more realistic model of inflation. 

\section{Geometric Set-up}\label{setup}
We require the general assumptions of \cite{Creminelli:2019pdh} with the cosmological constant replaced by a potential $U$ and the symmetry assumptions of \cite{creminelli2020sitter}. The former endows certain existence results that enable studying the cosmology via mean curvature flow (MCF) while the latter partially reduces the problem to the study of 2-dimensional manifolds whose Ricci scalar is intricately connected to its Euler characteristic.  
\subsection{General Assumptions}\label{sec:generalassump}
The general geometric assumptions are:
\begin{enumerate}[label=(\Alph*)]
    \item (\textit{Initially expanding slice in cosmology})
    There is a ``cosmology'', which is defined as a connected 3 + 1 dimensional spacetime
$M^{(3+1)}$ with a compact Cauchy surface. This implies that the spacetime is topologically
$M^{(3)} \times \mathbb{R}$ where $M^{(3)}$ is a compact 3-manifold, and that it can be foliated by a family of topologically identical Cauchy surfaces $M_t$ \cite{gerochdomain}. We fix one such foliation, i.e. such a time function $t$, with $t \in [t_0, +\infty)$ and $\nabla t$ being a non-zero timelike past-directed vector field, and with associated lapse function $N$: $N^{-2} \coloneqq - \nabla_{\mu}t \nabla^{\mu} t$,
$N > 0$. We consider manifolds that are initially expanding everywhere. That is, there is
an initial slice $M_0$, where $K > 0$ everywhere, with $K$ being the mean curvature with
respect to the future pointing normal to $M_0$. In particular, such a slice exists if one has a
global crushing singularity in the past.
    \item (\textit{EFE})
    $M^{(3+1)}$ satisfies Einstein's field equation
    \begin{equation}\label{eq:einsteinfield}
        R_{\mu \nu} - \frac{1}{2}g_{\mu\nu}R =8\pi G_N T_{\mu \nu} \ ,
    \end{equation}
    where $R_{\mu \nu}$ is the Ricci tensor, $R$ is the scalar curvature, and $T_{\mu \nu}$ is the stress-energy tensor. Next, we also suppose $T_{\mu \nu}$ can be decomposed as
    \begin{equation}
    T_{\mu \nu} = \overline{T}_{\mu \nu} - U g_{\mu \nu} \ ,
    \end{equation}
    where $U$ is a potential satisfying $\Lambda_1 \leq U \leq \Lambda_2$ for constants $\Lambda_1, \Lambda_2$ and $\overline{T}_{\mu \nu}$ is the reduced stress-energy tensor.
    \item (\textit{Energy conditions}) $U \geq \Lambda_1 > 0$ so the potential $U$ is positive. The reduced stress-energy tensor $\rt_{\mu \nu}$ satisfies the Dominant Energy Condition (DEC) and the Strong Energy Condition (SEC). The DEC states that $-\overline{T}^{\mu}_{\;\nu} k^{\nu}$ is a future-directed timelike or null vector for any future-directed timelike vector $k^{\nu}$. The DEC implies the Weak Energy Condition (WEC), $T_{\mu \nu}k^{\mu}k^{\nu} \geq 0$ for all time-like vectors $k^{\mu}$. The SEC, in 3+1D, states $\left(\overline{T}_{\mu \nu} - \frac{1}{2}g_{\mu \nu}T\right)k^{\mu}k^{\nu} \geq 0$ for any future-directed timelike vector $k^{\mu}$.
    \item (\textit{Crushing singularities}) We presume that the only spacetime singularities are of the crushing kind given in Definition \ref{def:crushing} (thus, those with zero spatial volume):
    
    \begin{definition}\label{def:crushing}
     (Definition 1 in \cite{Creminelli:2019pdh}, c.f. Definitions 2.9 --- 2.11 in \cite{eardleytime}). A future crushing function $\tilde t$ is a globally defined function on $M^{(3+1)}$ such that on a globally hyperbolic neighborhood ${\cal{N}}\cap\{\tilde t>c_0\}$, $\tilde t$ is a Cauchy time function with range $c_0<\tilde t<+\infty$ ($c_0 \ge 0$ is a constant), and such that the level sets~$S_c=\{\tilde t=c\}$, with $c>c_0$, have mean curvature $\tilde K<-c$.\footnote{For example in a Schwarzschild-de Sitter spacetime in the standard coordinates, one could take $\tilde t$ to be a function of $r$ for $r$ close to 0, so the level sets $S_c$ would be $r={\rm const}$.}  We shall say that a cosmology has potential singularities only of the crushing kind if there is an open set ${\cal{N}}$ such that, the inverse of the lapse of the $t$ foliation $N^{-1}$ is bounded outside ${\cal{N}}$, and such that ${\cal{N}}$ contains a Cauchy slice and admits a future crushing function $\tilde t$ with respect to which $N^{-1}$ is bounded in $\{\tilde t \leq c\}$ for any given $c$.    
\end{definition}
    
    Physically, $\mathcal{N}$ corresponds to a subset of the interior of black holes and we are imposing the condition that all possible pathologies occur as $\tilde{t} \to + \infty$.
\end{enumerate}
Notably, our assumptions allow for an inflationary potential $0 < \Lambda_1 \leq U \leq \Lambda_2$ associated with the stress-energy tensor of an inflaton field plus the stress-energy tensor of residual matter $\left(T_{matter}\right)_{\mu \nu}$, assuming $\left(T_{matter}\right)_{\mu \nu}$ also satisfies the DEC and SEC. Such a stress energy tensor takes the form:
\begin{equation}\label{eq:inflatonstressenergy}
T_{\mu \nu} = \nabla_{\mu}\phi \nabla_{\nu}\phi - \frac{1}{2}\left(\nabla_{\rho}\phi\nabla^{\rho}\phi\right)g_{\mu \nu} - U(\phi) g_{\mu \nu} + \left(T_{matter}\right)_{\mu \nu} \ .
\end{equation}
Note that this more realistic inflationary stress-energy tensor is precluded by \cite{creminelli2020sitter} which models the inflationary potential by a positive cosmological constant. We will work with a general $T_{\mu \nu}$ before specializing to this particular one in analyzing the dynamics of the inflaton in \autoref{sec:rolling}.

Now, since $\tilde{t}$ is bounded on the compact initial surface $M_0$, we pick large enough $c_1 > c_0$ such that $M_0 \subset \{\tilde{t} \leq c_1\}$. Then, we can redefine a new time function on $M^{(3+1)}$, which we call $t$ from now on,  whose lapse is set to $1$ in the region $\{\tilde{t} \leq c_1\}$.
    
\subsection{Mean Curvature Flow}
The general geometric assumptions above enable us to study $M^{(3+1)}$ via the Mean Curvature Flow (MCF) of co-dimension one spacelike surfaces in a Lorentzian manifold. This is defined as the deformation of a slice as follows: \(y^{\mu}(x, \lambda)\) is, at each \(\lambda,\) a mapping between the initial spatial manifold \(M_{0}\) (parametrized by \(x\)) and the global spacetime, \(M_{0} \times\left[0, \lambda_{0}\right) \rightarrow M^{(3+1)}\). The evolution under \(\lambda\) is given by (see for instance \cite{eckerhuisken}):
\begin{equation}\label{eq:mcfequation}
\frac{d}{d \lambda} y^{\mu}(x, \lambda)=K n^{\mu}\left(y^{\alpha}\right) \ ,
\end{equation}
where \(n^{\mu}\) is the future-oriented vector orthonormal to the surface of constant \(\lambda\). We denote by $\ml$ the geometric image of $y(\cdot, \lambda)$ and $K(x, \lambda)$ the mean curvature at $y(x, \lambda)$ with respect to the future pointing normal to $\ml$. MCF in Lorentzian manifolds is endowed with many regularity properties, which are rare in the context of Riemannian manifolds. Importantly, the flow is globally graphical and the spatial volume of $\ml$ is non-decreasing $\left(\text{see e.g. } \cite{Creminelli:2019pdh,creminelli2020sitter}\right)$. In our case, two additional properties are: a) $\ml$ remains spacelike; in fact the volume form $\sqrt{h}$ of $\ml$ is non-decreasing but it would vanish anywhere where $\ml$ becomes null $\left(\text{see e.g. } \cite{klebansenatore}\right)$ and, b) $K > 0$ everywhere is preserved on $\ml$ (see e.g. \cite{gerhardtbook}, Proposition 2.7.1).

Our assumptions, restricted to the case where the potential $U$ is just a positive cosmological constant $\Lambda > 0$, were used in \cite{Creminelli:2019pdh} to prove the useful statements about the maximum of $K$ and the existence of the MCF for all $\lambda \in [0, \infty)$. As pointed out in pg. 11 of \cite{Creminelli:2019pdh}, modifications of their results lead to essentially the same consequences in the current situation. 

\begin{theorem}\label{thm:maxk}
(Bound on the Maximum of $K$). \cite{Creminelli:2019pdh} Let $\ml$ be hypersurfaces satisfying the MCF equation, in an interval $\left[\lambda_1, \lambda_{2}\right]$, inside the smooth $(3+1)$ dimensional Lorentzian manifold $M^{(3+1)}$ satisfying \autoref{eq:einsteinfield} with $T_{\mu \nu} = T_{\mu \nu}' - \Lambda g_{\mu \nu}$ for some $\Lambda > 0$ and $T_{\mu \nu}'$ fulfilling the SEC. Then, defining $K_m(\lambda) = \sup_{x \in M_0}K(x, \lambda)$, we have for $\lambda_2 \geq \lambda_1 \geq 0$,
\begin{equation}
K_{m}\left(\lambda_{2}\right) \leq K_{\Lambda}\left(1+C_{1}(\lambda_1) e^{-\frac{2}{3} K_{\Lambda}^{2} \left(\lambda_2 - \lambda_{1}\right)}\right) \ , 
\end{equation}
with $K_{\Lambda} \coloneqq 24 \pi G_N \Lambda$ and $C_{1}(\lambda_1)=\max \left( \frac{K_{m}(\lambda_1)}{K_{\Lambda}}-1,0\right)$. So the maximum of $K$, if larger than $K_{\Lambda},$ decays exponentially fast towards $K_{\Lambda}$ with a rate given $K_{\Lambda}$.
\end{theorem}
Applying the above to $\Lambda = \Lambda_2$ and $T_{\mu \nu}' = \overline{T}_{\mu \nu} + (\Lambda_2 - U)g_{\mu \nu}$ which satisfies\footnote{Indeed, for any future-directed timelike vector $k^{\nu}$, 
$$\left(T_{\mu \nu}' - \frac{1}{2}g_{\mu \nu}T'\right)k^{\mu}k^{\nu} = \left(\overline{T}_{\mu \nu} - \frac{1}{2}g_{\mu \nu}\overline{T}\right)k^{\mu}k^{\nu} - (\Lambda_2 -  U) k_{\mu}k^{\mu} \geq 0 \ ,$$
since $\rt_{\mu \nu}$ satisfies the SEC, $\Lambda_2 - U \geq 0$, and $k_{\mu}k^{\mu} < 0$ since $k^{\mu}$ is timelike.
}
the SEC, we obtain for values of $\lambda \geq 0$, at which the flow exists for all $x \in M_0$,
\begin{equation}\label{eq:kupperpointwise}
K_{m}\left(\lambda \right) \leq \kt \left(1+C_{1} e^{-\frac{2}{3} \kt^{2} \lambda}\right) \ , 
\end{equation}
where $C_{1}=\max \left( \frac{K_{m}(0)}{\kt}-1,0\right)$. For convenience, we also rewrite the above as 
\begin{equation}\label{eq:ksquaredupperpointwise}
K_m^2\left(\lambda \right)  \leq  \kt^2 + C_2 \kt^2 e^{-\frac{2}{3}\kt^2 \lambda} \ ,
\end{equation}
where $C_2 \coloneqq C_1(2+C_1)$.

Now, this upper bound on the maximum of $K$ then implies an upper bound on the height function defined by $u(x, \lambda) \coloneqq  t(y(x, \lambda))$ (i.e. $u = t|_{\ml}$ is the Cauchy time function restricted to $\ml$) and the avoidance of crushing singularities.

\begin{theorem}\label{thm:heightbound}
(Height bound). \cite{Creminelli:2019pdh} Let $c>c_{0},$ with $c_{0}$ given in \autoref{def:crushing} and $\lambda_0 > 0$. There exists a constant $C \geq 0$ (depending on $\sup_{\tilde{t} \leq c}N^{-1}$ and \\
 $\sup_{(x, \lambda) \in M_0 \times [0, \lambda_0)}K$) such that the following holds. 

Provided that $M_{\lambda}$ is in $\tilde{t}<c$ for all $\lambda \in\left[0, \lambda_{0}\right)$, we have
$$
u(x, \lambda) \leq \sup _{x \in M_0} u(x, 0)+C \lambda \ ,
$$
for all $x \in M_0$ and $\lambda \leq \lambda_{0}$.
\end{theorem}

\begin{theorem}\label{thm:avoidsingular}
(Avoidance of crushing singularities). \cite{Creminelli:2019pdh} If the initial surface of the flow has $K \geq 0,$ then $M_{\lambda}$ stays away from $a$ crushing singularity. More precisely, for any $c>c_{0}>0$ such that $M_{0}$ is in $\{\tilde{t}<c\}$, the flow remains in $\{\tilde{t}<c\}$.
\end{theorem}

Theorem \ref{thm:avoidsingular} implies that the flow remains smooth while Theorem \ref{thm:heightbound} implies that $\ml$ stays in a compact region for $\lambda \in [0, \lambda_0)$. Combined with the short-term existence result below, they imply that the MCF exists for all $x \in M_0$ and $\lambda \in [0, \infty)$.

\begin{theorem}\label{thm:shortterm existence}
(Short-term existence of MCF). \cite{ecker_1993} Let $M^{(3+1)}$ be a smooth $(3+1)$-dimensional Lorentzian manifold satisfying \autoref{eq:einsteinfield} with $T_{\mu \nu} = T_{\mu \nu}' - \Lambda g_{\mu \nu}$ for some $\Lambda > 0$ and $T_{\mu \nu}'$ satisfying the WEC. Let $M_{0}$ be a compact smooth spacelike hypersurface in $M^{(3+1)}$. Then there exists a unique family $\left\{M_{\lambda}\right\}$ of smooth compact spacelike hypersurfaces satisfying the MCF equation \eqref{eq:mcfequation}, in an interval $\left[0, \lambda_{0}\right)$ for some $\lambda_{0}>0$ and having initial data $M_{0}$. Moreover, if this family stays inside a smooth compact region of $M^{(3+1)}$ then the solution can be extended beyond $\lambda_0$.
\end{theorem}

Note that \autoref{thm:shortterm existence} is indeed applicable with $\Lambda = \Lambda_1$ and $T_{\mu \nu}' = T_{\mu \nu} + (\Lambda_1 - U)g_{\mu \nu}$ which satisfies\footnote{For any future-directed timelike vector $k^{\nu}$, let $w^{\mu} = -\rt^{\mu}_{\;\nu} k^{\nu}$ be causal and future-directed since $\rt_{\mu \nu}$ satisfies the DEC. Now, define $v^{\mu} = -T' {}^{\mu}_{\;\nu} k^{\nu}$ which is future-directed as
$$v^{\mu}k_{\mu} =  w_{\mu}k^{\mu} - (\Lambda_1 - U)k_{\mu}k^{\mu} < 0 \ , $$
where $w_{\mu}k^{\mu} < 0$ since $w^{\mu}, k^{\mu}$ are future directed, $\Lambda_1 - U \leq 0$, and $k_{\mu}k^{\mu} < 0$ since $k^{\mu}$ is timelike. Next, $v^{\mu}$ is also causal since
$$v_{\mu}v^{\mu} = w_{\mu}w^{\mu} - 2(\Lambda_1 -U)w_{\mu}k^{\mu} + (\Lambda_1 - U)^2 k_{\mu}k^{\mu} \leq 0 \ ,$$
where $w_{\mu}w^{\mu} \leq 0$ since $w^{\mu}$ is causal and also, $w_{\mu}k^{\mu} < 0$, $\Lambda_1 - U \leq 0$ and $k_{\mu}k^{\mu} < 0$ as before.} the DEC and hence the WEC.

Finally, applying \autoref{thm:avoidsingular} with $c = c_1 > c_0$ implies $M_{\lambda}$ always stays in $\{\tilde{t} \leq c_1 \}$ such that for purposes of MCF, the lapse can be treated as unity. Ultimately, the global existence result of this section enables us to study $M^{(3+1)}$ via MCF while the upper bound in \autoref{eq:kupperpointwise} will grant us control over geometric quantities.

\subsection{Symmetry Assumptions}\label{sec:symmetryassum}
We also assume that either of the symmetry assumptions of \cite{creminelli2020sitter} holds. These enable the control of geometric quantities by topological properties of orbit surfaces. 

\begin{enumerate}[label=(\Alph*)]
    \item \textit{(Simplified symmetry assumption)}
    There is a Lie group $G$ which acts on $M^{(3)}$ such that the induced action on $M^{(3+1)}$ is by isometries, and such that the orbits under $G$ are closed surfaces. Assume that the orbits of $G$ are two-sided (i.e, with trivial normal bundle).
    \item \textit{(General symmetry assumption)}. There is a Lie group $\tilde{G}$ which acts on some cover $\pi: \tilde{M}^{(3)} \rightarrow M^{(3)},$ such that the induced action on $\tilde{M}^{(4)} \coloneqq \tilde{M}^{(3)} \times \mathbb{R}$ is by isometries. Assume further that the orbits under $\tilde{G}$ are two dimensional complete (i.e. with no edges) surfaces, and that for each such orbit $\tilde{\Sigma},$ its projection $\Sigma\coloneqq\pi(\tilde{\Sigma})$ is a two-sided surface.
\end{enumerate}

Notably, the induced actions of $G$ in Assumption A and $\tilde{G}$ in Assumption B on $M^{(3+1)}$ both map $M_0$ and thus its flow $\ml$ to themselves, thus preserving level sets of $\lambda$. Furthermore, if $\Sigma$ is an orbit under $G$ in Assumption A or the projection of an orbit under $\tilde{G}$ in Assumption B, all intrinsic and extrinsic scalar quantities on $\Sigma$ (such as $K, \rtwo, \rthree, A_{\mu \nu} A^{\mu \nu}, \left|K_{i j}\right|^{2},|\nabla K|^{2},\left|\sigma_{i j}\right|^{2}$ in \autoref{sec:notation}) are constant along it. 

An example of $M^{(3)}$ allowed by our assumptions is sketched in Fig. \ref{example_geom}. We defer more concrete constructions to Section 3 of \cite{creminelli2020sitter}.

\begin{figure}[t]
\centering
\includegraphics[width=0.8\textwidth]{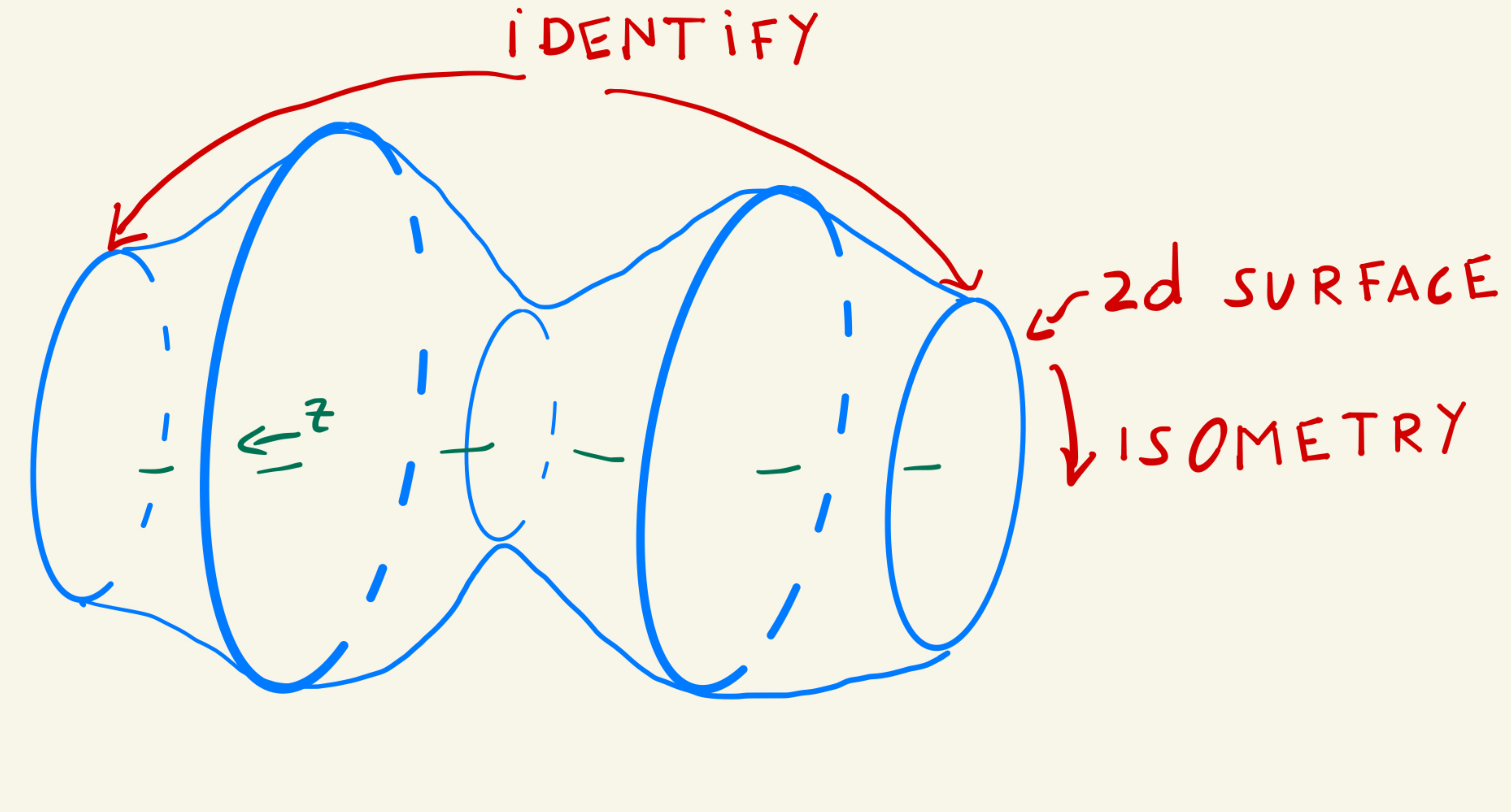}
\caption{Illustration of an example geometry permitted by our assumptions (reproduced with permission from \cite{creminelli2020sitter}). The circles represent (two-dimensional) orbits induced by the isometry group while $z$ is a parameter normal to the foliation by these orbits.}
\label{example_geom}
\end{figure}

Our symmetry assumptions also enable the following decomposition of the induced metric $g_{\mu \nu}$ on $\ml$ for a fixed $\lambda \geq 0$ \cite{creminelli2020sitter}. Consider the foliation of $\ml$ by the orbits of $G$ (or more generally, by the projections of the orbits of $\bar{G}$). By our two-sidedness assumption, there exists a global unit normal vector $E$ to this foliation. Let $z$ be the parameter along the flow lines of $E$, which is thus a signed distance function. Due to the isometries of $G$ (or $\tilde{G}$), the metric on $\mathcal{M}_{\lambda}$ has the warped product form
\begin{equation}\label{eq:spatialmetric}
g=d z^{2}+h_{z} \ ,
\end{equation}
where $h_{z}$ is a two-dimensional metric of constant curvature. 

By passing to a double cover, we can assume without loss of generality that the orbit surfaces $\Sigma$ are orientable. Thus, each such orbit is a two-dimensional orientable surface, with Euler characteristic $\chi =2,0,-2,-4, \ldots$. Later, this will be used to control the Ricci scalar on these orbits via the Gauss-Bonnet theorem.

\section{Notation and statement of main results}\label{sec:notation}
\subsection{Notation}

We will follow the notation in \cite{creminelli2020sitter}. The Riemann tensor is defined by $(\nabla_\mu\nabla_\nu - \nabla_\nu\nabla_\mu) \omega_\rho = R_{\mu\nu\rho}^{\quad\;\sigma} \omega_\sigma$, the Ricci tensor by $R_{\mu\nu} \coloneqq R_{\mu\sigma\nu}^{\quad\;\sigma}$ (we also use the notation $\mathrm{Ric}(a,b)$, with $a,b$ being two vectors), the Ricci scalar {(also known as scalar curvature)} by $R \coloneqq R_\mu^{\;\mu}$.

A time slice ${\cal M}_{\lambda}$ has an induced {Riemannian metric $g_{\mu\nu}$, and we can write  $g^{(4)}_{\mu\nu}=g_{\mu\nu}-n_\mu n_\nu$, where $g^{(4)}_{\mu\nu}$ is the spacetime metric (we use the mostly-plus convention) and  $n^\mu$} is orthonormal to ${\cal M}_\lambda$, $n_{\mu}n^{\mu} = -1$, and future-directed. The extrinsic curvature {(also known as second fundamental form)} of these slices is defined as $K_{\mu\nu} \coloneqq g_\mu^{\;\alpha} \nabla_\alpha n_\nu$, satisfying $n^{\mu} K_{\mu \nu} = 0$ and with trace (also known as mean curvature) $K \coloneqq g^{\mu\nu}K_{\mu\nu} = g^{(4)}{}^{\mu \nu} K_{\mu \nu}$, and traceless part $\sigma_{\mu \nu} \coloneqq K_{\mu \nu} - {1 \over 3} K h_{\mu \nu}$ (with our sign convention $K>0$ corresponds to expansion). We also define $\sigma^2 \coloneqq \sigma_{\mu\nu} \sigma^{\mu\nu}$ where $\sigma^2 \geq 0$ since $\sigma_{\mu\nu}$ is a tensor projected on the spatial hypersurfaces.  The Ricci tensor and Ricci scalar (scalar curvature) associated with the induced metric $g_{\mu\nu}$ on the $3$-dimensional slices are denoted, respectively, by  $\rthree_{\mu\nu}$ and $\rthree$.

Similarly, {each 2-dimensional symmetric orbit (or covering image of a symmetric orbit) $\Sigma$ within ${\cal M}_\lambda$ has induced metric {$h_{\mu\nu}$ satisfying $g_{\mu\nu}=h_{\mu\nu}+t_\mu t_\nu$, where $t^\mu$ is orthogonal to $\Sigma$ and to $n_{\mu}$, and $t_{\mu}t^{\mu} = 1$}. The extrinsic curvature (second fundamental form) of this slice within ${\cal M}_\lambda$  is defined as $A_{\mu\nu} \coloneqq h_\mu^{\;\alpha} \nabla_\alpha t_\nu$, satisfying $t^{\mu}A_{\mu \nu} = 0$ and with trace (mean curvature) $H \coloneqq h^{\mu\nu}A_{\mu\nu}$. The Ricci tensor and Ricci scalar (scalar curvature) associated with the induced metric $h_{\mu\nu}$ on a $2$-dimensional slices are denoted by $\rtwo_{\mu\nu}$ and $\rtwo$ respectively.

We denote by the capital letters  $C_i$ with $i=1,2,3,\ldots$, non-negative constants that depend only on the intrinsic and extrinsic properties of the initial 3-manifold of the flow: $M_{0}$.  We refer to such constants as universal.

\subsection{Main Results}
We summarize the main result of our paper as follows:

\begin{theorem}\label{thm:main_thm_intro}
Let $M^{(3+1)}$ be a spacetime satisfying assumptions $(A)-(D)$ of Section \ref{sec:generalassump}, in addition to  the symmetry assumptions of Section \ref{sec:symmetryassum}. If the orbit surfaces are spheres, one needs to further assume that the minimal area of an orbit surface in $M_0$ satisfies
\begin{equation}
S_{\min} \geq S_{\mathrm{lower}} \ ,
\end{equation}  
where $S_{\mathrm{lower}}$ depends only on $\max_{x\in M_0} K$, $\lo$, $\lt$. If $\frac{\Lambda_2}{\Lambda_1} < \frac{3}{2}$, for every $\delta > 0$, there exists some  $0\leq \tilde{\lambda}<\infty$ such that for $\lambda \geq \tilde{\lambda}$,
\begin{enumerate}[label=(\Alph*)]
\item \label{itm:volume_growth} (Volume grows between de Sitter spaces in the flat slicing with cosmological constants $\lo$ and $\lt$) Letting $V(\lambda)$ denote the volume of $\ml$,
\begin{equation}
(1+\delta)^{-1}e^{\ko^2 \left(\lambda - \tilde{\lambda}\right)} \leq \frac{V(\lambda)}{V(\tilde{\lambda})} \leq (1+\delta)e^{\kt^2 \left(\lambda - \tilde{\lambda}\right)} \ . 
\end{equation}
\item\label{itm:probe_time} (Majority of the physical volume reaches arbitrarily late times). Let $u: M_0 \times [0, \infty) \to \mathbb{R}$ be the height function defined by $u(x, \lambda) = t(y(x, \lambda))$. In other words, $u = t|_{\ml}$. Without loss of generality, also suppose $\min_{x \in M_0}u(x, \tilde{\lambda}) = 0$. For any $C(\lambda) > 0$, denoting $V(\{u \leq C(\lambda)\}, \lambda)$ as the physical volume of the subset of $\ml$ with $u \leq C(\lambda)$, we have
\begin{equation}
V(\{u \leq C(\lambda) \}, \lambda) \leq V(\tilde{\lambda})e^{(1+\delta)C(\lambda) \kt} \ .
\end{equation}
Combined with the lower bound in \ref{itm:volume_growth}, this implies for all $\varepsilon > 0$,
\begin{equation}
\lim_{\lambda \to \infty} \frac{V\left(\left\{u \leq (1 + \varepsilon)^{-1} \frac{\ko^2}{\kt} \lambda \right\}, \lambda \right)}{V(\lambda)} = 0 \ ,
\end{equation}
such that the majority of physical volume has height growing at least linearly in $\lambda$. We remark here that due to \autoref{eq:kupperpointwise}, $\limsup_{\lambda \to \infty} \frac{u(x, \lambda)}{\lambda} \leq \kt$ so the height function also grows at most linearly in $\lambda$.
\end{enumerate}
\end{theorem} 

\autoref{thm:main_thm_intro}\ref{itm:volume_growth} shows that there exists a family of spacelike surfaces $\{M_{\lambda} \}_{\lambda \geq 0}$ with volume exponentially expanding between those of flat slicings of de Sitter spaces with cosmological constants $\Lambda_1$ and $\Lambda_2$. In particular, inflationary expansion occurs in the absence of homogeneous initial conditions in this case, though an everywhere positive mean curvature on the initial slice $M_0$ is still assumed . 

Next, the exponentially expanding volume of $M_{\lambda}$ renders $(x, \lambda) \in M_0 \times [0, \infty)$ as an ideal ``MCF coordinate chart'' for modelling cosmological inflation. To this end, \autoref{thm:main_thm_intro}\ref{itm:probe_time} shows that this coordinate patch indeed covers a large subset of the cosmology and is hence useful as a physical description of the future.

The proof of \autoref{thm:main_thm_intro} (and more) occupies the upcoming two sections. In \autoref{sec:asymptotics_geometrical} we study the asymptotic behavior of the minimal area of orbit surfaces, length of the {transverse} geodesics $\gamma_{\lambda}$ on $\ml$, and the volume of $\ml$ (\autoref{thm:main_thm_intro}\ref{itm:volume_growth} ). \autoref{thm:main_thm_intro}\ref{itm:probe_time} and the asymptotics of the average height $u$ on $M_{\lambda}$ are covered in \autoref{sec:probe}. Finally, \autoref{sec:auxiliary} discusses several applications of the MCF coordinate chart concerning the variation of metric in the z-direction, the $L^1$ dilution of the reduced stress-energy tensor, the asymptotic evolution of an inflaton as well as directions for future work.

\section{Asymptotic Geometric Bounds}\label{sec:asymptotics_geometrical}
In this section, we follow the approach of \cite{creminelli2020sitter} to derive asymptotic bounds for various geometric quantities. The main idea is to repeatedly reduce the dimension of the problem until a two-dimensional orbit surface $\Sigma$, whose $\rtwo$ can be upper bounded in terms its Euler characteristic $\chi$. To this end, the Gauss equation will be essential in relating the Ricci tensor of a submanifold to that of its parent manifold. Applying the traced Gauss equation to $\ml$ within $M^{(3+1)}$, we obtain (see for instance \cite{Wald:1984rg}, Eq. $(\mathrm{E} .2 .27))$
\begin{equation}\label{eq:rawgauss3d}
R+2 \operatorname{Ric}(n, n)= \rthree -K_{\mu \nu} K^{\mu \nu}+K^{2}= \rthree -\sigma^{2}+\frac{2}{3} K^{2} \ . 
\end{equation}
Now, substituting $(n, n)$ into \autoref{eq:einsteinfield} yields
\begin{equation}
Ric(n, n) + \frac{1}{2} R = 8 \pi G_N \left(\rt(n, n) + U \right) \ ,
\end{equation}
such that \autoref{eq:rawgauss3d} becomes 
\begin{equation}
\rthree -\sigma^{2}+\frac{2}{3} K^{2}=16 \pi G_{N}\left(\rt(n, n)+ U\right)=16 \pi G_{N} \rt(n, n)+\frac{2}{3} K_{U}^{2} \ , 
\end{equation}
where 
\begin{equation}
K_U^2 \coloneqq 24\pi G_N U \ .
\end{equation}
In coordinate form:
\begin{equation}\label{eq:fullgauss3d}
\rthree +\frac{2}{3} K^{2}-\sigma^{2}=\frac{2}{3} K_{U}^{2}+16 \pi G_{N} \rt_{\mu \nu} n^{\mu} n^{\nu} \ ,
\end{equation}
where since $\rt_{\mu \nu} n^{\mu} n^{\nu} \geq 0$ by the WEC and using $K_U^2 \geq \ko^2 \coloneqq 24\pi G_N \Lambda_1 $, 
\begin{equation}\label{eq:gauss3d}
\rthree +\frac{2}{3} K^{2}-\sigma^{2} \geq \frac{2}{3} K_U^2 \geq \frac{2}{3}\ko^2  \ .
\end{equation}
Meanwhile, the traced Gauss equations for a surface orbit $\Sigma$ within $\ml$ (see \cite{creminelli2020sitter} Eq. (39)) imply
\begin{equation}\label{eq:gauss2d}
\rthree = \rtwo +2 \rthree_{z z}+A_{\mu \nu} A^{\mu \nu}-H^{2} \ . 
\end{equation}
In reducing the problem to lower dimensions, Eqns. \eqref{eq:fullgauss3d} and \eqref{eq:gauss2d} relate Ricci scalars on $M^{(3+1)}$ and $\ml$ as well as between $\ml$ and $\Sigma$ respectively.
\subsection{Orbit surfaces}
The main objective of this section is an upper bound for $\rtwo$ on any orbit surface $\Sigma$. Denote $S(z, \lambda)$ as the area of a two-dimensional orbit $\Sigma$ at coordinate $z$, defined in \autoref{eq:spatialmetric}, and $S_{min}(\lambda) \coloneqq \inf_{z} S(z, \lambda)$ as the minimal surface area of an orbit on $\ml$. From the Gauss-Bonnet formula, $\rtwo(z, \lambda) S(z, \lambda) = 4 \pi \chi$ which implies
\begin{equation}\label{eq:r2gaussbonnet}
\rtwo(z, \lambda) \leq \frac{4 \pi \chi_0}{S_{min}(\lambda)} \ ,
\end{equation}
where $\chi_0 \coloneqq \max(\chi, 0)$. Now, we show that if $\frac{\lt}{\lo} < \frac{3}{2}$,  $S_{\min}(\lambda)$ grows between the two-dimensional spatial slices of de Sitter spaces in the {flat slicing} with  cosmological constants $\Lambda_1, \Lambda_2$, up to a correction that approaches zero as $\frac{\Lambda_2}{\Lambda_1} \to 1$ (\footnote{The time evolution of topological surfaces in a similar setup was also studied in~\cite{Mirbabayi:2018suw}.}). This exponential growth of $S_{min}(\lambda)$ will imply that the upper bound for $\rtwo$ in \autoref{eq:r2gaussbonnet} decays exponentially if $\chi_0$ is positive. 

To prove the above claim, we will consider a differential equation for $S_{min}(\lambda)$. This is facilitated by the below standard lemma which shows that a minimizer has a well-defined derivative almost everywhere and obeys the fundamental theorem of calculus.

\begin{lemma}\label{lemma:hamilton}
(Hamilton's trick (c.f \cite{mantegazza} Lemma 2.1.3)). Let $f: K \times[a, b] \rightarrow \mathbb{R}$ be a smooth function with $K$ being compact, and set $g$ to be the minimizer of $f$ on $K$:
\begin{equation}
g(t)=\min _{x \in K} f(x, t)\ .
\end{equation}
Then $g$ is a Lipschitz function, and thus, differentiable almost everywhere and obeying the fundamental theorem of calculus. Moreover, if $t_{0}$ is a point of differentiability of $g$, and if $x_{0}$ is such that $f\left(x_{0}, t_{0}\right)=g\left(t_{0}\right)$ then
\begin{equation}
g'\left(t_{0}\right)=\left.\frac{\partial f}{\partial t}\right|_{\left(x_{0}, t_{0}\right)} \ .
\end{equation}
\end{lemma}

Next, we derive a few useful lemmas for our estimates. \autoref{lemma:exponenterror} enables us to wrap up various bounds once we have shown the logarithmic derivative of a quantity is a constant up to an exponentially decaying correction.
\begin{lemma}\label{lemma:exponenterror}
Let $f(\lambda)$ be a non-negative function that is differentiable almost everywhere and satisfies the fundamental theorem of calculus. Furthermore, suppose that at points of differentiability, there are constants $C, \varepsilon$ with $\varepsilon > 0$ such that
\begin{equation}\label{eq:updiffinequality}
\frac{d}{d\lambda} \log f = C + Err(\lambda)\ ,
\end{equation}
where $0 \leq Err(\lambda)  = O\left(e^{-\varepsilon \lambda}\right)$. Then for every $\delta > 0$, there exists large enough $\lambda_0$ such that for any $\lambda_2 \geq \lambda_1 \geq \lambda_0$, 
\begin{equation}\label{eq:upinequality}
\frac{f(\lambda_2)}{f(\lambda_1)} \leq (1 + \delta) e^{C(\lambda_2 -  \lambda_1)} \ . 
\end{equation}
Similarly, if
\begin{equation}
\frac{d}{d\lambda} \log f = C - Err(\lambda) \ ,
\end{equation}
where $0 \leq Err(\lambda)  = O\left(e^{-\varepsilon \lambda}\right)$, then for every $\delta > 0$, there exists large enough $\lambda_0$ such that for any $\lambda_2 \geq \lambda_1 \geq \lambda_0$, 
\begin{equation}
\frac{f(\lambda_2)}{f(\lambda_1)} \geq (1 + \delta)^{-1} e^{C(\lambda_2 -  \lambda_1)} \ .
\end{equation}
\end{lemma}
\begin{proof}
We only prove the upper bound as the proof works almost verbatim for the lower bound. Since $Err(\lambda) = O(e^{-\varepsilon \lambda})$, for any $\delta > 0$, there exists large enough $\lambda_0$ such that $\int_{\lambda_0}^{\infty}d\lambda \, Err(\lambda) \leq \ln(1 + \delta)$. Thus, integrating both sides of \autoref{eq:updiffinequality} from $\lambda_1$ to $\lambda_2$, 
\begin{equation}
\log \frac{f(\lambda_2)}{f(\lambda_1)} = C(\lambda_2 - \lambda_1) + \int_{\lambda_1}^{\lambda_2} d\lambda Err(\lambda) \leq  C(\lambda_2 - \lambda_1)  + \ln(1+\delta) \ .
\end{equation}
Upon exponentiation, we obtain \autoref{eq:upinequality}.

\end{proof}

The next two technical lemmas are necessary due to the looseness in the range of the inflationary potential $\Lambda_1 \leq U \leq \Lambda_2$.

\begin{lemma}\label{lemma:optimisationraw}
Let $f (x, y) = x^2 - cxy$ for $c > \sqrt{\frac{3}{2}}$. The optimization problem
$$\text{\normalfont minimize } f(x, y) \ , $$
subject to the constraints (with $a, b$ being some constants with $b \geq a > 0$)
$$x^2 - \frac{3}{2}y^2 \geq a^2 \ , $$
$$x \leq b \ , $$
$$0 \leq y \ , $$
has the solution $(x^*, y^*) =  \left(b, \sqrt{\frac{2}{3}\left(b^2 - a^2\right)}\right)$ with $f(x^*, y^*)  =b^2 - cb\sqrt{\frac{2}{3}\left(b^2 - a^2\right)} $.

\end{lemma}
\begin{proof}
Firstly, observe that the first constraint implies $x \geq a > 0$. Now, we argue that the minimizer $(x^*, y^*)$ of $f$ must occur on the curve $x^2 - \frac{3}{2}y^2 = a^2$. Otherwise $y^* < \sqrt{\frac{2}{3}\left(a^2 - (x^*)^2\right)}$ so $\left(x^*, \sqrt{\frac{2}{3}\left(a^2 - (x^*)^2\right)}\right)$ would result in a smaller value of $f(x, y)$ since $f$ is decreasing in $y$.

Next, we show that there is no local minimum of $f$ in the interior of $S = \{ x^2 - \frac{3}{2}y^2 = a^2\} \cap \{x \leq b\} \cap \{0 \leq y \}$. Otherwise, let $c(x, y) = x^2 - \frac{3}{2}y^2 - a^2 $ and suppose the minimizer $(x^*, y^*)$ is an interior point. By the method of Lagrange multipliers, we must have $\nabla c(x^*, y^*) \cdot \nabla f(x^*, y^*) = 0$ or
\begin{equation}
2x^*(2x^*-cy^*) + 3y^* \cdot cx^* = 4(x^*)^2 + cx^*y^* = 0 \ , 
\end{equation}
which is a contradiction since $x^*, y^*, c > 0$ ($y^* > 0 \in \text{int } S$). Hence, the minimizer $(x^*, y^*)$ must occur at $\partial S$ which are the points $(a, 0)$ and $\left(b, \sqrt{\frac{2}{3}\left(b^2 - a^2\right)}\right)$, at which $f$ takes respective values $a^2$ and $b^2 - cb\sqrt{\frac{2}{3}\left(b^2 - a^2\right)}$. Some algebra then shows $b^2 - cb\sqrt{\frac{2}{3}\left(b^2 - a^2\right)} < a^2$ for $c > \sqrt{\frac{3}{2}}$ so the minimizer is $\left(b, \sqrt{\frac{2}{3}\left(b^2 - a^2\right)}\right)$.
\end{proof}

For a slightly technical reason, we need to modify \autoref{lemma:optimisationraw} as follows.

\begin{lemma}\label{lemma:optimisation}
Suppose $|f - x^2| \leq cxy$ with $c > \sqrt{\frac{3}{2}}$ and $x, y$ satisfy the new constraints,
$$x^2 - \frac{3}{2}y^2 \geq a^2 - g \ ,$$
$$0 \leq x \leq b + \gamma \ , $$
$$0 \leq y \ , $$
$$0 \leq \gamma \leq C \ ,$$
where $a, b, C > 0$, $g \geq 0$ and $b \geq a$. Then,
\begin{equation}\label{eq:optimizationupper}
f \leq b^2 + cb\sqrt{\frac{2}{3}\left(b^2 - a^2\right)} + \varepsilon_1(\gamma) + \ep_1(g) \ ,
\end{equation}
\begin{equation}\label{eq:optimizationlower}
f \geq b^2 - cb\sqrt{\frac{2}{3}\left(b^2 - a^2\right)} - \varepsilon_2(\gamma)  - \ep_2(g) \ ,
\end{equation}
where
\begin{equation}
\varepsilon_1(\gamma) = 2b\gamma + \gamma^2 + c\gamma\sqrt{\frac{2}{3}(b^2 - a^2)} +  c(b+\gamma)\left(\sqrt{\frac{2}{3}\left(2b \gamma + \gamma^2 \right)} \right) \ ,
\end{equation}
\begin{equation}
\varepsilon_2(\gamma) = c\gamma\sqrt{\frac{2}{3}(b^2 - a^2)} +  c(b+\gamma)\left(\sqrt{\frac{2}{3}\left(2b \gamma + \gamma^2 \right)} \right) \ , 
\end{equation}
\begin{equation}
\ep_1(g) = c(b+C) \sqrt{\frac{2}{3}g} \ , 
\end{equation}
\begin{equation}
\ep_2(g) = c\sqrt{\frac{2}{3}\left(\left(b+C\right)^2 - a^2\right)} \sqrt{g} + c (b+ C + \sqrt{g}) \sqrt{\frac{2}{3}\left(2\left(b+C\right) \sqrt{g} + g \right)} \ ,
\end{equation}
are non-negative increasing functions with $\varepsilon_1(\gamma), \varepsilon_2(\gamma) = O\left(\gamma^{\frac{1}{2}}\right)$ as $\gamma \to 0$ while $\ep_1(g) = O\left(g^{\frac{1}{2}}\right)$ and $\ep_2(g) = O\left(g^{\frac{1}{4}}\right)$ as $g \to 0$.
\end{lemma}
\begin{proof}
Firstly, we consider the following auxiliary problem. Let $f_{\pm}(x, y) = x^2 \pm cxy$ with $c > \sqrt{\frac{3}{2}}$ and suppose (with $a, b> 0$, $g \geq 0$ being some constants and $b \geq a$)
$$x^2 - \frac{3}{2}y^2 \geq a^2 - g \ , $$
$$0 \leq x \leq b \ , $$
$$0 \leq y \ .$$
Then, we claim that
\begin{equation}\label{eq:optimizationplus}
f_+ \leq b^2 + cb\sqrt{\frac{2}{3}\left(b^2 - a^2\right)} + cb \sqrt{\frac{2}{3}g} \ , 
\end{equation}
\begin{equation}\label{eq:optimizationminus}
f_- \geq b^2 - cb \sqrt{\frac{2}{3}\left(b^2 - a^2\right)} -  c\sqrt{\frac{2}{3}\left(b^2 - a^2\right)} \sqrt{g} - c (b+\sqrt{g}) \sqrt{\frac{2}{3}\left(2b \sqrt{g} + g \right)} \ .
\end{equation}
To prove these, observe that $f_+$ is individually increasing in both $x, y$ under the constraints. Hence, the maximum of $f_+$ occurs at $(x^*, y^*) = \left(b, \sqrt{\frac{2}{3}\left(b^2 - a^2 + g\right)}\right)$. 
\begin{equation}
f_+ \leq f_+(x^*, y^*) = b^2 + c b \sqrt{\frac{2}{3}\left(b^2 - a^2 + g \right)} \leq b^2 + cb\sqrt{\frac{2}{3}\left(b^2 - a^2\right)} + cb \sqrt{\frac{2}{3}g} \ ,
\end{equation}
where we have used the inequality $\sqrt{d + e} \leq \sqrt{d} + \sqrt{e}$ for $d, e \geq 0$. 

Next, to bound $f_-$ from below, define $\tilde{x} = x + \sqrt{g}$ and $\tilde{f}_- = \tilde{x}^2 - c \tilde{x}y$. Then, according to the hypotheses, we have the bounds
$$\tilde{x}^2 - \frac{3}{2}y^2 \geq a^2 \ , $$
$$\tilde{x} \leq b + \sqrt{g} \ ,$$
$$0 \leq y \ ,$$
such that \autoref{lemma:optimisationraw} implies
\begin{equation}
\tilde{f}_- \geq (b+ \sqrt{g})^2 - c\left(b + \sqrt{g}\right)\sqrt{\frac{2}{3}\left( \left(b + \sqrt{g}\right)^2 - a^2 \right)} \ .
\end{equation}
Finally, using $\tilde{f}_- \leq x^2 + 2\sqrt{g}x + g - c xy \leq f_- + 2b \sqrt{g} + g$, and the inequality $\sqrt{d + e} \leq \sqrt{d} + \sqrt{e}$ again, we have
\begin{equation}
f_- \geq b^2 -  c\left(b + \sqrt{g}\right)\left(\sqrt{\frac{2}{3}\left(b^2 - a^2\right)} + \sqrt{\frac{2}{3}\left(2b\sqrt{g} + g \right)}\right) \ ,
\end{equation}
and upon expanding the terms in brackets, we obtain \autoref{eq:optimizationminus}.

The original claim in \autoref{lemma:optimisation} then stems from $f_- \leq f \leq f_+$ and substituting $b+\gamma$ for $b$ in the $b^2 \pm cb \sqrt{\frac{2}{3}\left(b^2 - a^2\right)}$ terms in Eqs. (\ref{eq:optimizationplus}) and (\ref{eq:optimizationminus}) as well as $b+C$ for $b$ in $cb\sqrt{\frac{2}{3}g}$ in \autoref{eq:optimizationplus} and the $-  c\sqrt{\frac{2}{3}\left(b^2 - a^2\right)} \sqrt{g} - c (b+\sqrt{g}) \sqrt{\frac{2}{3}\left(2b \sqrt{g} + g \right)}$ terms in \autoref{eq:optimizationminus} (then using the inequality $\sqrt{d + e} \leq \sqrt{d} + \sqrt{e}$ again).
\end{proof}

Equipped with these lemmas, we obtain the following bounds for $S_{min}(\lambda)$.
\begin{theorem}\label{thm:smingrowth}
Denote by $S_{\min }(\lambda)$ the minimal area of a surface orbit and $\chi$ its Euler characteristic. Suppose $\frac{\Lambda_2}{\Lambda_1} < \frac{3}{2}$. Then, for every $\delta > 0$, if either $\chi \leq 0,$ or, if $\chi=2$ and also $S_{\min }(0) \geq S_{\text {lower }}(\lo, \lt, C_1)$ where $C_1 = \max\left(\frac{K_m(0)}{\kt} - 1, 0\right)$ as in \autoref{eq:kupperpointwise}, there exists $\lambda_{0,1} \geq 0$ such that for all $\lambda_{0,1} \leq \lambda_{1}<\lambda_{2}$.
\begin{equation}\label{eq:smingrowth}
(1 + \delta)^{-1}e^{ \left( \frac{2}{3}\kt^2 - \sqrt{2}K_{cor}^2 \right)\left(\lambda_{2}-\lambda_{1}\right)} \leq \frac{S_{\min }\left(\lambda_{2}\right)}{S_{\min }\left(\lambda_{1}\right)} \leq (1+ \delta ) e^{\left(\frac{2}{3} \kt^{2} + \sqrt{2}K_{cor}^2 \right)\left(\lambda_{2}-\lambda_{1}\right)} \ ,
\end{equation}
where
\begin{equation}
K_{cor}^2 \coloneqq \kt \sqrt{ \frac{2}{3}\left(\kt^2 - \ko^2 \right)} \ .
\end{equation}
\end{theorem}
\begin{remark}
The condition $\frac{\Lambda_2}{\Lambda_1} < \frac{3}{2}$ ensures that $\frac{2}{3}\kt^2 - \sqrt{2}K_{cor}^2 > 0$. Technically, \autoref{eq:smingrowth} holds for $\chi < 0$ even without this condition, though the lower bound will not be as meaningful.
\end{remark}

\begin{proof}
Firstly, we show that on a minimal surface orbit, $\rthree \leq \rtwo $. By the Riccatti equation (primes indicate derivative with respect to $z$), 
\begin{equation}\label{eq:ricatti}
H' + A_{\mu \nu}A^{\mu \nu} = - \rthree_{zz} \ .
\end{equation}
Now, substituting \autoref{eq:ricatti} into \autoref{eq:gauss2d} to eliminate $\rthree_{zz}$, we obtain
\begin{equation}\label{eq:ricattiricci}
\rthree = \rtwo - A_{\mu \nu}A^{\mu \nu} - H^2 - 2H' \ .
\end{equation}
Recall that by the first variation of area formula, $\lie_{\partial_z} \sqrt{h} = H \sqrt{h}$ where $\sqrt{h}$ is the volume element of the surface orbit. Hence, on the minimal $z$-slice, $H = 0$ and $H' \geq 0$ such that \autoref{eq:ricattiricci} implies
\begin{equation}\label{eq:minrthreertwo}
\rthree \leq \rtwo \ .
\end{equation}
Substituting this into \autoref{eq:gauss3d} and applying \autoref{eq:r2gaussbonnet}, yields on a minimal slice, 
\begin{equation}\label{eq:klowerminsurf}
K^2 - \frac{3}{2}\sigma^2 \geq \ko^2 - \frac{6\pi \chi_0}{S_{min}} \ .
\end{equation}
Next, we obtain a differential equation for the evolution of $S_{min}(\lambda)$. The evolution of the metric under MCF is given by (see \cite{eckerhuisken}, Proposition 3.1)
\begin{equation}\label{eq:metricevolution}
\frac{d g_{ij}}{d \lambda} = 2KK_{ij} = \frac{2}{3}K^2g_{ij} + 2K \sigma_{ij} \ .
\end{equation}
Also, note that in general, if $x_{1}, x_{2}$ are local coordinates around a point on a orbit surface $\Sigma$ at flow time $\lambda$ (not necessarily of minimal area), we have

\begin{equation}
\frac{d}{d \lambda} \log \sqrt{\operatorname{det}^{12} g_{i j}}=\frac{1}{2} \operatorname{tr}^{(12)}\left(g^{i j} \frac{d g_{j k}}{d \lambda}\right)= \frac{2}{3} K^{2}+K g^{i j} \sigma_{j i} \ ,
\end{equation}
where the indices $i, j$ only run over $1, 2$ (corresponding to the tangent vectors $\partial_{x_1}$ and $\partial_{x_2}$). This will be used in the relationship (where $S$ is the area of any orbit)
\begin{equation}\label{eq:sareaevolution}
\frac{dS}{d \lambda}=\int_{\Sigma} d^{2} x \frac{d}{d \lambda}\left(\sqrt{\operatorname{det}^{12} g_{i j}}\right)=\int_{\Sigma} d^{2} x \sqrt{\operatorname{det}^{12} g_{i j}}\left(\frac{2}{3} K^{2}+K g^{i j} \sigma_{ji}\right) \ ,
\end{equation}
which is applicable even to $S_{\min }(\lambda)$ by Hamilton's trick (\autoref{lemma:hamilton}). Hence, we are left with bounding the integrand on the R.H.S of \autoref{eq:sareaevolution} on a minimal $z$-slice. For any point on the $z$-slice, we can without loss of generality assume that $x_1, x_2$ are local orthonormal coordinates such that $|g^{ij}\sigma_{ji}| = |\sigma_{11} + \sigma_{22}| \leq |\sigma_{11}| + |\sigma_{22}| \leq \sqrt{2 \sigma^2}$ by the triangle and Cauchy-Schwartz inequalities. Thus,
\begin{equation}
\left|\left(\frac{2}{3} K^{2}+K g^{i j} \sigma_{j i}\right) - \frac{2}{3}K^2\right| \leq \sqrt{2}K\sqrt{\sigma^2} \ ,
\end{equation}
such that multiplying the above by $\frac{3}{2}$ and applying \autoref{lemma:optimisation} to $f = K^2 + \frac{3}{2}Kg^{ij}\sigma_{ji}$, $x = K$, $y = \sqrt{\sigma^2}$ and $c = \frac{3}{2}\sqrt{2}$,  with $a = \ko, g = \frac{6\pi \chi_0}{S_{min}}$ in light of \autoref{eq:klowerminsurf}, $b = \kt, \gamma = \kt C_1 e^{-\frac{2}{3}\kt^2 \lambda}, C = \kt C_1$ due to \autoref{eq:kupperpointwise}, 
\begin{equation}
\frac{2}{3} K^{2}+K g^{i j}\sigma_{ji} \geq \frac{2}{3}\kt^2 - \sqrt{2}K_{cor}^2  -  \tilde{\varepsilon}_2(\lambda) - \tilde{\ep}_2\left(\frac{6\pi \chi_0}{S_{min}}\right) \ ,
\end{equation}
where $\tilde{\varepsilon}_2(\lambda) = \frac{2}{3}\varepsilon_2\left(\kt C_1 e^{-\frac{2}{3}\kt^2 \lambda}\right) = O\left(e^{-\frac{1}{3}\kt^2 \lambda}\right)$ is decreasing in $\lambda$ and $\tilde{\ep}_2\left(\frac{6\pi \chi_0}{S_{min}}\right) = \frac{2}{3}\ep_2\left(\frac{6\pi \chi_0}{S_{min}}\right)$, with $\varepsilon_2$ and $\ep_2$ defined as in \autoref{lemma:optimisation}. Combining this with \autoref{eq:sareaevolution} and \autoref{lemma:hamilton}, 
\begin{equation}\label{eq:lowerperturb}
\frac{d}{d\lambda} \log S_{min} \geq   \frac{2}{3}\kt^2 - \sqrt{2}K_{cor}^2 -  \tilde{\varepsilon}_2(\lambda)  - \tilde{\ep}_2\left(\frac{6\pi \chi_0}{S_{min}}\right) \ .
\end{equation}
If $\chi < 0$ such that $\chi_0 = 0$ and $\tilde{\ep}_2\left(\frac{6\pi \chi_0}{S_{min}}\right) = 0$, \autoref{lemma:exponenterror} immediately implies the lower bound of \autoref{eq:smingrowth} since $\tilde{\varepsilon}_2(\lambda) = O\left(e^{-\frac{1}{3}\kt^2 \lambda}\right)$. Otherwise if $\chi = 2$, there exists large enough $\lambda_{0, 1}'(\lo, \lt, C_1)$ such that $\tilde{\varepsilon}_2(\lambda) \leq \frac{1}{3}\left(\frac{2}{3}\kt^2 - \sqrt{2}K_{cor}^2\right)$ for $\lambda \geq \lambda_{0, 1}'$. Since $\tilde{\ep}_2$ is an increasing function and $\tilde{\ep}_2\left(\frac{6\pi \chi_0}{S_{min}}\right) \to 0$ as $\frac{6\pi \chi_0}{S_{min}} \to 0$, there exists a large enough $\hat{S}$ such that as long as $S_{min} \geq \hat{S}$, $\tilde{\ep}_2\left(\frac{6\pi \chi_0}{S_{min}}\right) < \frac{1}{3}\left(\frac{2}{3}\kt^2 - \sqrt{2}K_{cor}^2\right)$ so 
\begin{equation}\label{eq:exponentialconstantsmin}
\frac{d}{d\lambda} \log S_{min}  \geq \frac{2}{3}\left(\frac{2}{3}\kt^2 - \sqrt{2}K_{cor}^2\right) - \tilde{\varepsilon}_2 (\lambda) \geq - C\kt^2 \ ,
\end{equation}
where $C\kt^2 \coloneqq -\frac{2}{3}\left(\frac{2}{3}\kt^2 - \sqrt{2}K_{cor}^2\right) + \tilde{\varepsilon}_2(0)$ (recall $\tilde{\varepsilon}_2(\lambda)$ is decreasing). As a result, we get that if $S_{min}(0) \geq S_{lower}$ with
\begin{equation}\label{eq:defslower}
S_{lower}(\lo, \lt, C_1) \coloneqq \hat{S}e^{C \kt^2 \lambda_{0, 1}'}
\end{equation}
then $S_{min}(\lambda) \geq \hat{S}$ on $[0, \lambda_{0, 1}']$. Furthermore, we see that for all $\lambda \geq \lambda_{0, 1}'$ such that $S_{min}(\lambda) \geq \hat{S}$,
\begin{align}
\frac{d}{d\lambda} \log S_{min} (\lambda) &\geq \frac{2}{3}\left(\frac{2}{3}\kt^2 - \sqrt{2}K_{cor}^2\right) - \tilde{\varepsilon}_2 (\lambda) \nonumber \\
&\geq \frac{1}{3}\left(\frac{2}{3}\kt^2 - \sqrt{2}K_{cor}^2 \right) \tag{$\lambda \geq \lambda_{0, 1}'$ and our choice of $\lambda_{0, 1}'$}  \\
&> 0 \ , \label{eq:non-optimal}
\end{align}
which shows that if $S_{min}(\lambda) \geq \hat{S}$ on $[0, \bar{\lambda}]$ with $\bar{\lambda} \geq \lambda_{0, 1}'$, we also have  $S_{min}(\lambda) \geq \hat{S}$ on some extended interval $[0, \bar{\lambda} + \delta]$ for some $\delta > 0$. {To see this, observe that the R.H.S of $\autoref{eq:lowerperturb}$ is $>0$ at $\lambda = \bar{\lambda}$ as in \autoref{eq:non-optimal}. } The continuity of $S_{min}(\lambda)$ (since it is Lipschitz by \autoref{lemma:hamilton}) and of the R.H.S of \autoref{eq:lowerperturb} in $\lambda$ and in $S_{min}$ then implies that there exists some $\delta > 0$ such that $\frac{d}{d\lambda}S_{min}(\lambda) > 0$ for points of differentiability in $(\bar{\lambda} - 2 \delta, \bar{\lambda} + 2 \delta)$ such that for any $\lambda \in [\bar{\lambda}, \bar{\lambda} + \delta]$, $S_{min}(\lambda) = S_{min}(\bar{\lambda}) + \int_{ \bar{\lambda}}^{\lambda} d\lambda' \frac{d}{d\lambda} S_{min}(\lambda ' ) \geq S_{min}(\bar{\lambda}) \geq \hat{S}$.

{Thus, letting $\lambda^* = \sup_{\lambda \geq \lambda_{0, 1}'} \{\lambda: S_{min}(\lambda') \geq \hat{S} \, \,  \forall  \lambda' \in [0, \lambda] \}$, we must have $\lambda^* = \infty$ otherwise by continuity $S_{min}(\lambda^*) = \hat{S}$ such that the extension property above implies that $S_{min} \geq \hat{S}$ on $[0, \lambda^* + \delta]$, violating the supremum property of $\lambda^*$. }Hence,  $S_{min}(\lambda) \geq \hat{S}$ for all $\lambda \geq 0$ such that Eq. \eqref{eq:non-optimal} holds for all $\lambda \geq \lambda_{0, 1}'$.

Then, we have the non-optimal estimate $S_{min} = \Omega \left(e^{\frac{1}{3}\left(\frac{2}{3}\kt^2 - \sqrt{2}K_{cor}^2 \right) \lambda}\right) $ upon integrating \autoref{eq:non-optimal} such that $\tilde{\ep}_2\left(\frac{6\pi \chi_0}{S_{min}}\right) = O\left(e^{-\frac{1}{12}\left(\frac{2}{3}\kt^2 - \sqrt{2}K_{cor}^2 \right)\lambda}\right)$. Recalling $\tilde{\varepsilon}_2(\lambda) = O\left(e^{-\frac{1}{3}\kt^2 \lambda}\right)$, both $\tilde{\varepsilon}_2$ and $\tilde{\ep}_2$ terms in \autoref{eq:lowerperturb} are exponentially decaying in $\lambda$ so \autoref{lemma:exponenterror} implies the lower bound of \autoref{eq:smingrowth} for $\chi = 2$ if $S_{min}(0) \geq S_{lower}$.

Similarly, from the upper bound of \autoref{lemma:optimisation} (with the same substitutions for $f, x, y, c, a, g, b, \gamma, C$ as before),
\begin{equation}
\frac{2}{3} K^{2}+K g^{i j} \leq \frac{2}{3}\kt^2 - \sqrt{2}K_{cor}^2  +  \tilde{\varepsilon}_1(\lambda) + \tilde{\ep}_1 \left(\frac{6\pi \chi_0}{S_{min}}\right) \ ,
\end{equation}
where $\tilde{\varepsilon}_1(\lambda) = \frac{2}{3}\varepsilon_1\left(\kt C_1 e^{-\frac{2}{3}\kt^2 \lambda} \right) =  O\left(e^{-\frac{1}{3}\kt^2 \lambda}\right)$ and $\tilde{\ep}_1\left(\frac{6\pi \chi_0}{S_{min}}\right) = \frac{2}{3} \ep_1\left(\frac{6\pi \chi_0}{S_{min}}\right)$. Combining this with \autoref{eq:sareaevolution} and \autoref{lemma:hamilton}, 
\begin{equation}\label{eq:upperperturb}
\frac{d}{d\lambda} \log S_{min} \leq  \frac{2}{3}\kt^2 - \sqrt{2}K_{cor}^2  +  \tilde{\varepsilon}_1(\lambda) + \tilde{\ep}_1 \left(\frac{6\pi \chi_0}{S_{min}}\right) \ .
\end{equation}
Now if $\chi_0 = 0$ so $\tilde{\ep}_1 \left(\frac{6\pi \chi_0}{S_{min}}\right) = 0$, \autoref{lemma:exponenterror} immediately implies the upper bound of \autoref{eq:smingrowth} since $\tilde{\varepsilon}_1(\lambda) = O\left(e^{-\frac{1}{3}\kt^2 \lambda}\right)$. Otherwise if $\chi = 2$ and $S_{min}(0) \geq S_{lower}$, the lower bound of \autoref{eq:smingrowth} shows $\tilde{\ep}_1 \left(\frac{6\pi \chi_0}{S_{min}}\right) = O\left(e^{-\frac{1}{2}\left(\frac{2}{3}\kt^2 - \sqrt{2}K_{cor}^2 \right)\lambda}\right)$. Hence, \autoref{lemma:exponenterror} again implies the upper bound of \autoref{eq:smingrowth}.
\end{proof}

Notice that the additional requirement in the case of the sphere depends exponentially on
the initial conditions (see \autoref{eq:defslower}). This is different from what happens in the case of complete homogeneity where, for Bianchi-IX universes, one has to impose a lower bound on $\rthree$ \cite{Wald:1983ky}. This bound however does not depend exponentially on the initial conditions.

Henceforth, we choose $\lambda_{0, 1}$ in \autoref{thm:smingrowth} for a fixed $\delta = 1$.

\begin{corollary}
Suppose the conditions of \autoref{thm:smingrowth} hold. Substituting the lower bound of \autoref{eq:smingrowth} into \autoref{eq:r2gaussbonnet}, we have for $\lambda_2 \geq \lambda_1 \geq \lambda_{0, 1}$,
\begin{equation}\label{eq:r2pointwise}
\rtwo(z, \lambda_2) \leq \frac{8 \pi \chi_0}{S_{min}(\lambda_{1})}e^{-\left(\frac{2}{3}\kt^2 - \sqrt{2}K_{cor}^2\right)(\lambda_2 - \lambda_{1})} \ .
\end{equation}
\end{corollary}

\subsection{Transverse Length and Volume}
In this section, we derive bounds for the asymptotic growth of the volume $V(\lambda)$ and ``tranverse'' length $L(\lambda)$ of $\ml$ where the latter is defined as follows. By the form of the metric in \autoref{eq:spatialmetric}, it is straightforward to check that if a geodesic is at a point tangent to the vector $E$ (unit normal to the foliation of $\ml$ by orbit surfaces), it is tangent to $E$ everywhere.  Denote therefore by $L(\lambda)$ the length, at time $\lambda$, of any geodesic $\gamma$  that is parallel to the $z$-direction, from an initial orbit surface to itself.

Firstly, we derive a lemma for the asymptotic bounds of averages of $K^2$. This is important as the mean curvature $K$ controls the growth of various geometric quantities under MCF. 

Below, $dV_{\lambda} = \sqrt{g(\lambda)} d^3 x$ denotes the volume form on $\ml$. Furthermore, for any smooth function $f$, we denote $\langle f \rangle(\lambda) \coloneqq \frac{1}{V(\lambda)} \intmv f$ and  $\langle f \rangle_z (\lambda) \coloneqq \frac{1}{L(\lambda)} \intmz f$ .

\begin{lemma}\label{lemma:avgksquared}
We have
\begin{equation}\label{eq:avgksquaredupper}
\langle K^2 \rangle , \langle K^2 \rangle_z  \leq \kt^2 \left(1 + C_2 e^{-\frac{2}{3}K_{\Lambda_2}^2 \lambda}\right) \ .
\end{equation}
Furthermore, under the conditions of \autoref{thm:smingrowth}, for $\lambda_2 \geq \lambda_1 \geq \lambda_{0, 1}$, 
\begin{equation}\label{eq:avgksquaredlowerwithu}
\left\langle K^2 - K_U^2 - \frac{3}{2}\sigma^2 \right\rangle, \left\langle K^2 - K_U^2 - \frac{3}{2}\sigma^2 \right\rangle_z \geq -\frac{12\pi \chi_0}{S_{min}(\lambda_{1})}e^{-\left(\frac{2}{3}\kt^2 - \sqrt{2}K_{cor}^2\right)(\lambda_2 - \lambda_{1})} \ .
\end{equation}
In particular, since $K_U^2 \geq \ko^2$,
\begin{equation}\label{eq:avgksquaredlower}
\left\langle K^2 - \frac{3}{2}\sigma^2 \right\rangle, \left\langle K^2 - \frac{3}{2}\sigma^2 \right\rangle_z \geq \ko^2 - \frac{12\pi \chi_0}{S_{min}(\lambda_{1})}e^{-\left(\frac{2}{3}\kt^2 - \sqrt{2}K_{cor}^2\right)(\lambda_2 - \lambda_{1})} \ .
\end{equation}
\end{lemma}
\begin{proof}
The upper bounds are consequences of the pointwise bound in \autoref{eq:ksquaredupperpointwise}.
For the lower bound, integrating \autoref{eq:ricattiricci} from $z =0$ to $z = L(\lambda)$ and noting that $\intmz H' = 0$ since the endpoints refer to the same $z$-coordinate, we obtain
\begin{equation}\label{eq:r2r3zdir}
\intmz \rthree \leq \intmz \rtwo \ .
\end{equation}
Integrating \autoref{eq:gauss3d} from $z =0$ to $z = L(\lambda)$ and applying the above inequality,
\begin{equation}
\intmz \left(K^2 - K_U^2 - \frac{3}{2}\sigma^2\right) \geq -\frac{3}{2} \intmz \rtwo \ .
\end{equation}
Hence, by \autoref{eq:r2pointwise}, we obtain upon dividing by $L(\lambda)$,
\begin{equation}
\left\langle K^2 - K_U^2 - \frac{3}{2}\sigma^2 \right\rangle_z \geq -\frac{12\pi \chi_0}{S_{min}(\lambda_{1})}e^{-\left(\frac{2}{3}\kt^2 - \sqrt{2}K_{cor}^2\right)(\lambda_2 - \lambda_{1})} \ .
\end{equation}
The same proof works for $\left\langle K^2 - K_U^2 - \frac{3}{2}\sigma^2 \right\rangle$, by replacing integrals above by integrals over $\ml$.
\end{proof}

The next theorem describes the asymptotic growth of $V(\lambda)$ and $L(\lambda)$. In particular, $V(\lambda)$ and $L(\lambda)$ grow between what are expected of de Sitter spaces in the {flat slicing} with  cosmological constants $\Lambda_1, \Lambda_2$, with the latter $L(\lambda)$ having a correction term that approaches zero as $\frac{\Lambda_2}{\Lambda_1} \to 1$. 

\begin{theorem}\label{thm:vollengthgrowth}
Under the conditions of \autoref{thm:smingrowth}, for any $\delta > 0$, there exists $\lambda_{0, 2} \geq \lambda_{0, 1}$ such that for all $\lambda_2 \geq \lambda_1 \geq \lambda_{0, 2}$, 
\begin{equation}\label{eq:volumegrowth}
(1+\delta)^{-1}e^{\ko^2(\lambda_2 - \lambda_1)} \leq \frac{V(\lambda_2)}{V(\lambda_1)} \leq (1+\delta)e^{\kt^2(\lambda_2 - \lambda_1)} \ ,
\end{equation}
\begin{equation}\label{eq:lengthgrowth}
(1+\delta)^{-1}e^{\left(\frac{1}{3}\kt^2-K_{cor}^2\right) (\lambda_2 - \lambda_1)} \leq \frac{L(\lambda_2)}{L(\lambda_1)} \leq (1+\delta)e^{\left(\frac{1}{3}\kt^2+K_{cor}^2\right)(\lambda_2 - \lambda_1)} \ .
\end{equation}
\end{theorem}
\begin{remark}
If $\frac{\lt}{\lo} < \frac{6}{5}$, then $\frac{1}{3}\kt^2-K_{cor}^2 > 0$ so \autoref{eq:lengthgrowth} implies $L(\lambda)$ grows exponentially.
\end{remark}

\begin{proof}
From $\frac{dV}{d\lambda} = \intmv K^2$ such that $\frac{d}{d\lambda}\log V = \langle K^2 \rangle$, we obtain for $\lambda \geq \lambda_{0, 1}$,
\begin{equation}\label{eq:volumeperturb}
\ko^2 -\frac{12\pi \chi_0}{S_{min}(\lambda_{0, 1})}e^{-\left(\frac{2}{3}\kt^2 - \sqrt{2}K_{cor}^2\right)(\lambda - \lambda_{0, 1})} \leq \frac{d}{d\lambda} \log V  \leq \kt^2 \left(1 +C_2e^{-\frac{2}{3}\kt^2 \lambda} \right) \ , 
\end{equation}
due to \autoref{lemma:avgksquared}. Applying \autoref{lemma:exponenterror}, we obtain \autoref{eq:volumegrowth}. Next, we compute
\begin{equation}
\frac{dL}{d\lambda} = \intmz KK_{zz} = \intmz \left(\frac{K^2}{3} + K\sigma_{zz}\right) \ .
\end{equation}
By the triangle and Cauchy-Schwartz inequalities,
\begin{equation}
 \left|\intmz K \sigma_{zz} \right| \leq \left( \intmz K^2 \right)^{\frac{1}{2}} \left( \intmz \sigma^2 \right)^{\frac{1}{2}} \ .
\end{equation}
Expressing the two equations above in terms of averages along the $z$-direction, we have
\begin{equation}
\frac{d}{d \lambda} \log L = \left\langle \frac{K^2}{3} + K\sigma_{zz}\right\rangle_z \ ,
\end{equation}
\begin{equation}
\left|\left\langle K\sigma_{zz}\right\rangle_z\right| \leq \left\langle K^2 \right\rangle_z^{\frac{1}{2}} \left\langle \sigma^2 \right\rangle_z^{\frac{1}{2}} \ ,
\end{equation}
which imply
\begin{equation}\label{eq:lemma4logL}
\left|3\frac{d}{d\lambda}\log L - \left\langle K^2 \right\rangle_z \right| \leq 3 \left\langle K^2 \right\rangle_z^{\frac{1}{2}} \left\langle \sigma^2 \right\rangle_z^{\frac{1}{2}} \ .
\end{equation}
Furthermore from \autoref{lemma:avgksquared}, we have for $\lambda \geq \lambda_{0, 1}$ (where we have used the inequality $\sqrt{d+e} \leq \sqrt{d} + \sqrt{e}$ in the second equation below)
\begin{equation}
\left\langle K^2 \right\rangle_z - \frac{3}{2}\left\langle \sigma^2 \right\rangle_z \geq \ko^2 - \frac{12 \pi \chi_0}{S_{min}(\lambda_{0, 1})}e^{-\left(\frac{2}{3}\kt^2 - \sqrt{2}K_{cor}^2\right)(\lambda - \lambda_{0, 1})} \ ,
\end{equation}
\begin{equation}
\left\langle K^2 \right\rangle_z^{\frac{1}{2}} \leq \kt + \kt \sqrt{C_2} e^{-\frac{1}{3}\kt^2 \lambda} \ ,
\end{equation}
which correspond to the first two constraints in \autoref{lemma:optimisation} with $x = \left\langle K^2 \right\rangle_z^{\frac{1}{2}}$, $y = \left\langle \sigma^2 \right\rangle_z^{\frac{1}{2}}$, $c = 3$, $a = \ko$, $b = \kt$, $\gamma = \kt \sqrt{C_2} e^{-\frac{1}{3}\kt^2 \lambda}$, $C = \kt \sqrt{C_2}$, $g = \frac{12\pi \chi_0}{S_{min}(\lambda_{0, 1})}e^{-\left(\frac{2}{3}\kt^2 - \sqrt{2}K_{cor}^2\right)(\lambda - \lambda_{0, 1})}$. In light of \autoref{eq:lemma4logL}, using $f = 3\frac{d}{d\lambda}\log L $ in \autoref{lemma:optimisation} ,
\begin{align}\label{eq:lengthperturb}
\frac{\kt^2}{3} - K_{cor}^2  - O&\left(\max\left(e^{-\frac{1}{4}\left(\frac{2}{3}\kt^2 - \sqrt{2}K_{cor}^2\right)\lambda},e^{-\frac{1}{6}\kt^2 \lambda}\right) \right) \leq  \frac{d}{d\lambda} \log L \nonumber \\
&\leq \frac{\kt^2}{3} + K_{cor}^2 +  O\left( \max\left(e^{-\frac{1}{2}\left(\frac{2}{3}\kt^2 - \sqrt{2}K_{cor}^2\right)\lambda}, e^{-\frac{1}{6}\kt^2 \lambda}\right) \right) \ ,
\end{align}
such that \autoref{lemma:exponenterror} implies \autoref{eq:lengthgrowth}.
\end{proof}

\subsection{Flow Reset}
For a fixed $\delta > 0$ in \autoref{thm:vollengthgrowth}, let $\lambda_{0, 3} \coloneqq \max\left(\lambda_{0, 2}, \frac{3}{2\kt^2}\ln \frac{C_1}{\delta} \right)$ where the latter condition is such that $K_m(\lambda_{0,3}) \leq (1+\delta)\kt$ by \autoref{eq:kupperpointwise}. Now, for the ease of notation, let us re-define the initial time of the flow as to be $\lambda_{0,3}$, so from now on $\lambda_{0,3}=0$. Note that estimates \eqref{eq:kupperpointwise} and \eqref{eq:ksquaredupperpointwise} still hold (in fact, with $C_1 = \delta$ and $C_2 = \delta(2+\delta)$). \\

In particular, we have, for every $\lambda \geq 0$,
\begin{equation}\label{eq:redkupperpointwise}
K(x, \lambda) \leq \left(1 + \delta e^{-\frac{2}{3}\kt^2 \lambda}\right)\kt \leq (1+\delta) \kt \ , 
\end{equation}
\begin{equation}\label{eq:redminarea}
\frac{1}{2}e^{\left(\frac{2}{3}\ko^2 - \sqrt{2}K_{cor}^2\right)\lambda} \leq \frac{S_{min}(\lambda)}{S_{min}(0)} \leq 2 e^{\left(\frac{2}{3}\ko^2 + \sqrt{2}K_{cor}^2\right)\lambda} \ , 
\end{equation}
\begin{equation}\label{eq:redvolume}
(1+\delta)^{-1}e^{\ko^2 \lambda} \leq \frac{V(\lambda)}{V(0)} \leq (1+\delta) e^{\kt^2 \lambda} \ ,
\end{equation}
\begin{equation}\label{eq:redlength}
(1+\delta)^{-1} e^{\left(\frac{1}{3}\ko^2 - K_{cor}^2\right)\lambda} \leq \frac{L(\lambda)}{L(0)} \leq (1+\delta) e^{\left(\frac{1}{3}\ko^2 + K_{cor}^2\right)\lambda} \ ,
\end{equation}
\begin{equation}\label{eq:redrtwopointwise}
\rtwo(x, \lambda) \leq C_3\kt^2 e^{-\left(\frac{2}{3}\kt^2 - \sqrt{2}K_{cor}^2\right)\lambda} \ ,
\end{equation}
\begin{equation}\label{eq:redavgksquared}
\left\langle K^2 - K_U^2 - \frac{3}{2}\sigma^2 \right\rangle, \left\langle K^2 - K_U^2 - \frac{3}{2}\sigma^2 \right\rangle_z \geq - \frac{3}{2}C_3\kt^2 e^{-\left(\frac{2}{3}\kt^2 - \sqrt{2}K_{cor}^2\right)\lambda} \ ,
\end{equation}
where $C_3\kt^2 \coloneqq \frac{ 8 \pi \chi_0}{S_{min}(0)}$.

\section{MCF Probes Large Regions of Cosmology}\label{sec:probe}
The exponentially expanding volume of $M_{\lambda}$ renders $\lambda$ an ideal time coordinate for modelling inflation (when combined with the spatial coordinate $x$ on $M_0$). However, $\lambda$ can only be a meaningful coordinate if the slices $\{ M_{\lambda} \}_{\lambda \geq 0}$ probe a large subset of the cosmology. Physically, the main possible pathology is the flow slowing down to a halt, as it encounter black holes in the cosmology (such that $K$ approaches zero), in an effort to avoid them. This is not a defect of MCF per say but, by design, a mechanism that enables us to study smooth regions of the cosmology (expected to dominate on the large-scale), without having to explicitly handle singularities. That said, we do have to check that the flow does not trivially stop everywhere and is able to reach large regions of the cosmology.

In the context of a positive cosmological constant $\Lambda$, \cite{creminelli2020sitter} achieved this by proving that $\{M_{\lambda} \}_{\lambda \geq 0}$ is in fact a foliation of the future $M^{(3+1)} \cap \{ t \geq 0 \}$. Inadvertently, this implies that the cosmology does not have any crushing singularities and is thus geodesically complete (since the only spacetime singularities are assumed to be of the crushing kind). Their result thus precludes black holes expected to be physically relevant.

For a more general inflationary potential $\lo \leq U \leq \lt$, we hereby prove weaker statements which show that $\{ M_{\lambda} \}_{\lambda \geq 0}$ probes a large subset of the future. Concretely, we show that almost all of the total physical volume on $M_{\lambda}$ has height $u = t|_{\ml}$ that approaches $+\infty$ as $\lambda \to + \infty$ and, in fact the heights of these points grow at least as fast as $\frac{\ko^2 }{\kt} \lambda$. Though we do not know if black holes indeed form, our results do not explicitly preclude black holes in the cosmology.




Below, we assume without loss of generality that $u_{min}(0) \coloneqq \inf_{x \in M_0}u(x, 0) = 0$ (otherwise one can redefine such a $t$ by subtracting a constant). Furthermore, we use the notation $V(A, \lambda) = \intmv \mathbbm{1}_{A}$ for any subset $A \subset \ml$.

\begin{theorem}\label{thm:smallheightvolume}
For any $\lambda \geq 0$ and $C(\lambda)$, 
\begin{equation}
V(\{u \leq C(\lambda) \}, \lambda) \leq V(0) e^{(1+\delta) C(\lambda) \kt} \ .
\end{equation}
As a result, for every $\varepsilon > 0$,
\begin{equation}\label{eq:heightlineargrowth}
\lim_{\lambda \to \infty} \frac{V\left(\left\{u \leq (1 + \varepsilon)^{-1} \frac{\ko^2}{\kt} \lambda \right\}, \lambda \right) }{V(\lambda)} = 0 \ ,
\end{equation}
which shows that almost all of the physical volume of $M_{\lambda}$ has height growing at least as fast as $\frac{\ko^2 }{\kt}\lambda$ asymptotically.

\end{theorem}

\begin{proof}
Firstly, observe that
\begin{equation}\label{eq:heightincrease}
\partial_{\lambda} u(x, \lambda) = g^{(4)}\left(\nabla t, K n \right) \geq K |\nabla t| |n| = K \ ,
\end{equation}
where we have applied the reverse Cauchy-Schwartz inequality for timelike vectors in the second step (note that $g^{(4)}\left(\nabla t, n \right) > 0$ since $\nabla t$ is past-directed while $n$ is future-directed). Hence, for $(x, \lambda)$ such that $u(x, \lambda) \leq C(\lambda)$, \autoref{eq:heightincrease} implies 
\begin{equation}
\int_0^{\lambda} d \lambda' K(x, \lambda') \leq C(\lambda) - u(x, 0) \leq C(\lambda) - u_{min}(0) = C(\lambda) \ ,
\end{equation}
such that since $K \leq \left(1+\delta \right) \kt$ for $\lambda \geq 0$, we have
\begin{equation}
\int_0^{\lambda} d \lambda' K^2(x, \lambda') \leq (1+\delta) \kt C(\lambda) \ ,
\end{equation}
and from the evolution equation $\frac{\partial}{\partial\lambda}\log \sqrt{g} = K^2$,
\begin{equation}
\sqrt{g(x, \lambda)} = \sqrt{g(x, 0)}e^{\int_0^{\lambda} d\lambda K^2(x, \lambda)} \leq \sqrt{g(x, 0)} e^{(1+\delta) \kt C(\lambda)} 
\end{equation}
at these points. Then,
\begin{align}
V(\{u \leq C(\lambda) \}, \lambda) &\leq \int_{M^{(3)}} \sqrt{g(x, 0)}e^{ (1+\delta) \kt C(\lambda) } d^3x \mathbbm{1}_{u(x, \lambda) \leq C}  \nonumber \\
&\leq e^{ (1+\delta) \kt C (\lambda)}  \int_{M^{(3)}} \sqrt{g(x, 0)} d^3x \nonumber \\
&=  V(0)e^{(1+\delta) \kt C(\lambda)} \ .
\end{align}
Substituting $C(\lambda) = (1 + 2 \delta)^{-1} \frac{\ko^2}{\kt} \lambda $ and recalling that $V(\lambda) = \Omega\left(e^{-\ko^2 \lambda} \right)$ in light of the lower bound in \autoref{eq:volumegrowth}, 
\begin{equation}
\lim_{\lambda \to \infty} \frac{V\left(\left\{u \leq (1 + 2\delta)^{-1} \frac{\ko^2}{\kt} \lambda \right\}, \lambda \right) }{V(\lambda)} = 0 \ .
\end{equation}
Finally, since $\delta$ is arbitrary, we obtain \autoref{eq:heightlineargrowth}. 
\end{proof}
We also have the following bounds on the max and average heights $u_{max}(\lambda) \coloneqq \max_{x \in M_0}u(x, \lambda)$ and $\langle u \rangle $ on $\ml$. 
\begin{corollary}
\begin{equation}\label{eq:uavgliminflimsup}
\frac{\ko^2}{\kt} \leq \liminf_{\lambda \to +\infty} \frac{\langle u \rangle (\lambda)}{\lambda} \leq \limsup_{\lambda \to +\infty} \frac{\langle u \rangle (\lambda)}{\lambda} \leq \limsup_{\lambda \to +\infty} \frac{u_{max} (\lambda)}{\lambda} \leq \kt \ .
\end{equation}
\end{corollary}
\begin{proof}
The only non-trivial inequalities are the left-most and the right-most ones. For the former, since $u$ is increasing in $\lambda$ due to \autoref{eq:heightincrease}, $u(x, \lambda) \geq u_{min}(0) = 0$ for any $x \in M_0$ and $\lambda \geq 0$. Notably, for any $\lambda > 0$ and $\varepsilon > 0$,
\begin{align}
\frac{\langle u \rangle (\lambda)}{\lambda} &\geq \frac{1}{\lambda V(\lambda)}  \Bigg[ (1 + \varepsilon)^{-1} \frac{\ko^2}{\kt}\lambda \cdot V\left(\left\{u  > (1 + \varepsilon)^{-1}  \frac{\ko^2}{\kt}\lambda \right\}, \lambda \right) \nonumber \\
&+ u_{min}(0) \cdot  V\left(\left\{u  \leq  (1 + \varepsilon)^{-1}  \frac{\ko^2}{\kt}\lambda \right\}, \lambda  \right) \Bigg] \nonumber \\
&= (1+ \varepsilon)^{-1} \frac{\ko^2}{\kt} \frac{V\left(\left\{u  > (1 + \varepsilon)^{-1}  \frac{\ko^2}{\kt}\lambda \right\}, \lambda \right) }{V(\lambda) } \nonumber \\
&=  (1+ \varepsilon)^{-1} \frac{\ko^2}{\kt}  \left(1 - \frac{V\left(\left\{u  \leq (1 + \varepsilon)^{-1}  \frac{\ko^2}{\kt}\lambda \right\}, \lambda \right) }{V(\lambda) } \right) \ .
\end{align}
Taking $\liminf_{\lambda \to + \infty}$ and applying \autoref{eq:heightlineargrowth} yields for all $\varepsilon > 0$,
\begin{equation}
\liminf_{\lambda \to + \infty} \frac{\langle u \rangle (\lambda)}{\lambda}  \geq (1 + \varepsilon)^{-1} \frac{\ko^2}{\kt} \ .
\end{equation}
Finally, taking $\varepsilon \to 0$ gives the left-most inequality. For the right-most inequality, \autoref{lemma:hamilton} implies
\begin{equation}\label{eq:umaxderivative}
\partial_{\lambda}u_{max} = g^{(4)}\left(\nabla t, K n \right) = K \leq (1+\delta)\kt \ ,
\end{equation}
where $\nabla t$ and $n$ in the second expression are parallel at the maximizer of $u$ (so $g^{(4)}\left(\nabla t, n \right) = 1$). Integrating \autoref{eq:umaxderivative}  yields $\limsup_{\lambda \to + \infty} \frac{u_{max}(\lambda)}{\lambda} \leq (1+ \delta) \kt$ after which taking $\delta \to 0$ yields the right-most inequality.
\end{proof}

\section{Auxiliary Results}\label{sec:auxiliary}
Having established that $\lambda$ is a useful time coordinate which asymptotically describes the future of the cosmology, we explore a few ramifications of using such a coordinate system.
\subsection{Variation of Metric in $z$-direction}\label{sec:metricz}
Observe that a priori, there might be large regions where $U = \lo$ and $U = \lt$ respectively. In these regions, the de Sitter no-hair theorem in \cite{creminelli2020sitter} suggests (but does not prove) that the metrics in these cases approach those of de Sitter with cosmological constants $\lo$ and $\lt$ respectively. Then, the spatial metric $g$ of $\ml$ could vary significantly across these regions as $\lambda$ grows. In this section, for large $\lambda$, we show that $g$ varies much more weakly along the $z$-direction than the extreme case described. In particular, our results suggest that such hypothetical de Sitter domains (if they exist) must be separated by a physical distance that tends to infinity as $\lambda \to \infty$, illustrating that a physical observer cannot experience two different cosmological constants at late enough times.

The propagation of the metric along the level sets is given by the second fundamental form (extrinsic curvature):
\begin{equation}\label{eq:propmetricz}
\lie_{\partial_z} g_{\mu \nu} = 2A_{\mu \nu} \ ,
\end{equation}
so we proceed by controlling $A_{\mu \nu}$. Firstly, \autoref{eq:gauss3d} implies the following pointwise bound on $\rthree$,
\begin{equation}\label{eq:rthreepointwise}
\rthree \geq \frac{2}{3}\left(\kt^2 - K^2\right) - \frac{2}{3}\left(\kt^2 - \ko^2\right) \geq -\frac{2}{3}C_2 \kt^2 e^{-\frac{2}{3}\kt^2 \lambda} - \frac{2}{3}\left(\kt^2 - \ko^2\right) \ ,
\end{equation}
where we have applied \autoref{eq:ksquaredupperpointwise} in the second step. Substituting the pointwise bounds for $\rtwo$ and $\rthree$ given by Eqs. \eqref{eq:redrtwopointwise} and \eqref{eq:rthreepointwise} into \autoref{eq:ricattiricci},
\begin{align}\label{eq:ricattiricciafter}
H^2 + 2H' + A_{\mu \nu} A^{\mu \nu} &= \rtwo - \rthree \nonumber \\
&\leq \frac{2}{3}\left(\kt^2 - \ko^2\right) + \left(C_3 + \frac{2}{3}C_2\right)\kt^2 e^{-\left(\frac{2}{3}\kt^2 - \sqrt{2}K_{cor}^2\right)\lambda} \ .
\end{align}
Using $A_{\mu \nu}A^{\mu \nu} \geq \frac{1}{2}H^2$ by the Cauchy-Schwartz inequality and noting that $H' = 0$ and the maximum and minimum of $H$, we obtain the following pointwise bound
\begin{align}\label{eq:hpointwise}
|H| &\leq \sqrt{\frac{4}{9}\left(\kt^2 - \ko^2 \right) + \frac{2}{3}\left(C_3 + \frac{2}{3}C_2\right)\kt^2 e^{-\left(\frac{2}{3}\kt^2 - \sqrt{2}K_{cor}^2\right)\lambda}} \nonumber   \\
&\leq \frac{2}{3}\sqrt{\kt^2 - \ko^2 } + \sqrt{ \frac{2}{3}\left(C_3 + \frac{2}{3}C_2\right) } \kt e^{-\frac{1}{2}\left(\frac{2}{3}\kt^2 - \sqrt{2}K_{cor}^2\right)\lambda} \nonumber \\
&\coloneqq \frac{2}{3}\sqrt{\kt^2 - \ko^2 }  + \varepsilon_{\lambda} \ .
\end{align}
Integrating \autoref{eq:ricattiricciafter} along the z-direction and using the pointwise bound \autoref{eq:hpointwise}, 
\begin{equation}
\int_0^z dz A_{\mu \nu}A^{\mu \nu} \leq \frac{2}{3}\left(\kt^2 - \ko^2\right)z + \frac{8}{3}\sqrt{\kt^2 - \ko^2 } + \frac{3}{2}\varepsilon_{\lambda}^2z  + 4\varepsilon_{\lambda} \ .
\end{equation}
Using the Cauchy-Schwartz equality and the inequality $\sqrt{d+e} \leq \sqrt{d} + \sqrt{e}$,
\begin{equation}\label{eq:aintegralbound}
\int_0^z dz |A| \leq \sqrt{ \frac{2}{3} \left[\left(\sqrt{\kt^2 - \ko^2 }z + 2 \right)^2 -4\right]} + \sqrt{\frac{3}{2}\varepsilon_{\lambda}^2 z^2 + 4\varepsilon_{\lambda}z} \ .
\end{equation}
Equipped with this bound, we proceed with bounding the variation in the metric. For any product coordinate system on $\ml$ of the form $(\alpha, \beta, z)$ where $z$ is as before and $\partial_{\alpha}, \partial_{\beta}$ are tangent to each surface orbit, \autoref{eq:propmetricz} implies
\begin{equation}
|\partial_z g_{\alpha \alpha}| = |2A_{\alpha \alpha}| \leq 2|A| g_{\alpha \alpha} \ .
\end{equation}

Integrating this from $z=z_1$ to $z=z_2$ and using \autoref{eq:aintegralbound}, with $\Delta z \coloneqq |z_2 - z_1|$, 
\begin{equation}\label{eq:logmetriczdiff}
\left| \log \frac{g_{\alpha \alpha}(z_2, \lambda)}{g_{\alpha \alpha}(z_1, \lambda)}\right| \leq   2\sqrt{ \frac{2}{3} \left[\left(\sqrt{\kt^2 - \ko^2 }\Delta z + 2 \right)^2 -4\right]} + 2\sqrt{\frac{3}{2}\varepsilon_{\lambda}^2 \left(\Delta z \right)^2 + 4\varepsilon_{\lambda}\Delta z} \ .
\end{equation}
One consequence of the above is as follows. Suppose there are two points $x_1, x_2 \in M_0$ such that under MCF, $g_{\alpha \alpha}$ at these points expand as in de Sitter with different cosmological constants $\Lambda_3 \neq \Lambda_4$ in the FLRW slicing with $\lo \leq \Lambda_3, \Lambda_4 \leq \lt$. Now, let $z_1(\lambda), z_2(\lambda)$ denote the respective $z$-coordinates of $x_1, x_2$ on $\ml$ (identified via MCF). Then since $g_{\alpha\alpha}(z_1(\lambda), \lambda) = \Theta\left(e^{\frac{1}{3}K_{\Lambda_3}^2 \lambda}\right)$and  $g_{\alpha\alpha}(z_2(\lambda), \lambda) = \Theta\left(e^{\frac{1}{3}K_{\Lambda_4}^2 \lambda}\right)$, \\ $\left| \log \frac{g_{\alpha \alpha}(z_2(\lambda), \lambda)}{g_{\alpha \alpha}(z_1(\lambda), \lambda)}\right| = \Theta \left( \left|K_{\Lambda_3}^2 - K_{\Lambda_4}^2\right| \lambda \right)$ which implies $\Delta z(\lambda) = \Omega \left( \frac{\left|K_{\Lambda_3}^2 - K_{\Lambda_4}^2\right|}{\sqrt{\kt^2 - \ko^2}}\lambda\right) $ by \autoref{eq:logmetriczdiff}. Notably, $\Delta z(\lambda) \to \infty$ as $\lambda \to \infty$, showing that these points must have a diverging physical separation in $z$.
\subsection{$L^1$ Dilution of Reduced Stress-Energy Tensor}\label{sec:l1stressenergy}
Next, we show that as a consequence of inflationary expansion, matter partially dilutes and the reduced stress-energy tensor is at most $\lt - \lo$ almost everywhere. We achieve this by bounding the average of $|\rt_{\mu\nu}n^\mu n^\nu| = \rt_{\mu\nu}n^\mu n^\nu$ (by the WEC) on a transverse geodesic. Note that
\begin{align}
 \frac{16\pi G_N}{L(\lambda)} &\intmz  \rt_{\mu\nu}n^\mu n^\nu \nonumber \\
&= \frac{1}{L(\lambda)}\intmz \left(\rthree + \frac{2}{3}\left(K^2 - K_U^2\right) - \sigma^2\right) \tag{\autoref{eq:fullgauss3d}}  \\
&\leq  \frac{1}{L(\lambda)}\intmz \left(\rtwo + \frac{2}{3}\left(K^2 - \kt^2\right) + \frac{2}{3}\left(\kt^2 - \ko^2 \right) \right) \tag{\autoref{eq:r2r3zdir}}  \\
&\leq \frac{2}{3}\left(\kt^2 - \ko^2 \right)  + \left(\frac{2}{3}C_2 + C_3 \right)\kt^2 e^{- \left(\frac{2}{3}\kt^2 - \sqrt{2} K_{cor}^2 \right)\lambda} \ ,
\end{align}
where we have applied Eqns. \eqref{eq:redrtwopointwise} and \eqref{eq:ksquaredupperpointwise} in the last inequality. Because of the DEC, $\rt_{\mu\nu} n^\mu n^\nu$ is at least as large as the absolute value of any other component of the stress tensor in an orthonormal frame where $n^\mu$ is the timelike vector\footnote{This is actually an equivalent definition of the DEC \cite{Hawking:1973uf} as it is straightforward to verify.}. We therefore define a vierbein $e_\mu{}^{a}$, such that $g_{\mu\nu}^{(4)}=e_\mu{}^{a}e_\nu{}^{b}\eta_{ab}$, with $\eta_{ab}$ being the Minkowski metric. We choose $e_\mu{}^0=n_\mu$. By the DEC, we have
\begin{align}
\frac{16\pi G_N}{L(\lambda)}&\intmz \left| \rt_{\mu\nu} e^{\mu a}e^{\nu b}\right|\\
&\leq \frac{16\pi G_N}{L(\lambda)} \intmz \rt_{\mu\nu}n^\mu n^\nu \nonumber \\
&\leq \frac{2}{3}\left(\kt^2 - \ko^2 \right) + \left(\frac{2}{3}C_2 + C_3 \right)\kt^2 e^{- \left(\frac{2}{3}\kt^2 - \sqrt{2} K_{cor}^2 \right)\lambda} \ .
\end{align}
Since, by the symmetries of the problem, $\rt_{\mu\nu}$ is uniform on the slices at constant $z$, we see that in almost-all of the ever-growing $z$-direction, $ G_N \rt_{\mu\nu}$ has to be at most on the order of $\frac{\kt^2 -\ko^2}{24 \pi G_N} = \lt - \lo$ up to a correction of order $ \left(\frac{2}{3}C_2 + C_3 \right)\lt$ on a z-distance that is at most an $O\left(e^{-\left(\frac{2}{3}\kt^2 - \sqrt{2}K_{cor}^2\right)\lambda}\right)$ fraction of the total length $L(\lambda)$. Physically, this suggests that for small enough $\frac{\Lambda_2 - \Lambda_1}{\Lambda_1} \ll 1$, the influences of matter and the inflaton kinetic energy are bounded and the potential component is the dominant factor in determining the asymptotic geometry.

\subsection{Asymptotic Average Rolling of Scalar Field in Monotonic Potential}\label{sec:rolling}
In this section, we devote our attention to the stress-energy tensor of an inflaton given by \autoref{eq:inflatonstressenergy}. Under the additional assumption that the potential $U$ is increasing (i.e. $U'(\phi) > 0$), we show that asymptotically, the average rolling of $\phi$ in MCF coordinates resembles that of a slow-rolling solution in an expanding FLRW universe. Given a point $p \in \ml$, the metric of $M^{(3+1)}$ at $p$ is given by
\begin{equation}\label{eq:metricmcfcoordinates}
ds^2 = g_{\mu \nu}^{(4)}dx^{\mu}dx^{\nu} = - K^2 d\lambda^2 + g_{ij}dx^i dx^j \ .
\end{equation}
This motivates our definition of the average rolling of $\phi$ on $\ml$ as $\langle K^{-1} \partial_{\lambda} \phi \rangle(\lambda) \coloneqq \frac{1}{V(\lambda)} \intmv K^{-1}\partial_{\lambda}\phi $. Now, the dynamics of $\phi$ is governed by the Klein-Gordon equation:
\begin{equation}\label{eq:kleingordonraw}
\Box \phi = U'(\phi) \ , 
\end{equation}
where the Laplace-Beltrami operator $\Box$ on $M^{(3+1)}$ is defined by
\begin{equation}\label{eq:laplacebeltrami}
\Box \phi = \frac{1}{\sqrt{-g^{(4)}}}\partial_{\mu}\left(\sqrt{-g^{(4)}} \partial^{\mu}\phi\right) = \frac{1}{K\sqrt{g}} \partial_{\mu}\left( K \sqrt{g} \partial^{\mu}\phi\right) \ .
\end{equation}
Here, we recall that $g$ is the spatial metric on $\ml$. Below, we will adopt the additional assumption that $U'(\phi) > 0$ (i.e. the potential is increasing). An example potential is depicted in Fig. \ref{examplepotential}.

\begin{figure}[!t]
\centering
\includegraphics[scale=0.5]{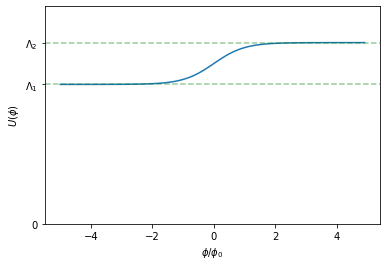}
\caption{Example of an increasing potential $\lo \leq U \leq \lt$ given by $U(\phi) = \frac{\lt - \lo}{2} \tanh\left(\frac{\phi}{\phi_0}\right) + \frac{\lt + \lo}{2}$ where $\phi_0 > 0$ is some inflaton scale.}
\label{examplepotential}
\end{figure}

\begin{theorem}\label{thm:slowrollasymptotic}
Let $\phi$ evolve via the Klein-Gordon equation \eqref{eq:kleingordonraw} with $U'(\phi) > 0$. Furthermore, suppose the conditions of \autoref{thm:smingrowth} hold. Then, 
\begin{equation}
\limsup_{\lambda \to \infty} \langle K^{-1}\parl \phi \rangle \leq 0 \ , 
\end{equation}
\begin{equation}
\liminf_{\lambda \to \infty} \langle K^{-1}\parl \phi \rangle \geq - \frac{\kt }{\ko^2} \limsup_{\lambda \to \infty} \left \langle U' \right \rangle \ .
\end{equation}
In particular, if $U'(\phi) \leq M$ is bounded, we have
\begin{equation}
\liminf_{\lambda \to \infty} \langle K^{-1}\parl \phi \rangle \geq -\frac{\kt M}{\ko^2} \ .
\end{equation}
\end{theorem}

\begin{proof}
Substituting \autoref{eq:laplacebeltrami} in \autoref{eq:kleingordonraw}, we obtain upon multiplying by $-K$ and integrating over $\ml$:
\begin{align}\label{eq:integratedkg}
-\intmv KU' &= -\int_{M^{(3)}} d^3 x \partial_{\mu}\left(K \sqrt{g} \partial^{\mu}\phi\right) \nonumber \\
&= -\frac{d}{d\lambda} \int_{M^{(3)}} d^3 x  K \sqrt{g} \partial^{\lambda} \phi \nonumber \\
&= \frac{d}{d\lambda} \intmv K^{-1} \partial_{\lambda} \phi \ ,
\end{align}
where in the second step, we have used Stokes' theorem along with the fact that $M^{(3)}$ has no boundary to conclude that the spatial derivatives vanish upon integration. Next, we compute
\begin{align}\label{eq:avgrollevolution}
\frac{d}{d\lambda}\langle K^{-1}\parl \phi \rangle  &= \frac{1}{V}\frac{d}{d\lambda} \intmv K^{-1} \partial_{\lambda} \phi - \frac{1}{V}\frac{dV}{d\lambda} \langle K^{-1}\parl \phi \rangle  \nonumber \\
&= - \left\langle K U' \right\rangle  -  \langle K^2 \rangle \langle K^{-1}\parl \phi \rangle \ ,
\end{align}
where we have applied \autoref{eq:integratedkg} and the equation $\frac{dV}{d\lambda} = \intmv K^2$ in the second step. Now since $- \langle K U' \rangle < 0$, by a standard ODE argument, if $\langle K^{-1}\parl \phi \rangle(\lambda^*) \leq 0$ for some $\lambda^*$, $\langle K^{-1}\parl \phi \rangle(\lambda) \leq 0$ for all $\lambda \geq \lambda^*$.

Concretely, suppose otherwise that $\langle K^{-1}\parl \phi \rangle(\lambda_0) > 0$ for some $\lambda_0 \geq \lambda^*$. Then, define $\bar{\lambda} = \inf_{\lambda \geq \lambda^*} \{\lambda:  \langle K^{-1}\parl \phi \rangle (\lambda)  > 0 \}$ such that by continuity, $\langle K^{-1}\parl \phi  \rangle(\bar{\lambda}) = 0$. Then, Eq. \eqref{eq:avgrollevolution} implies that $\frac{d}{d\lambda}\langle K^{-1}\parl \phi \rangle (\bar{\lambda}) = - \langle K U ' \rangle < 0$. Then, the continuity of the R.H.S of Eq. \eqref{eq:avgrollevolution} in $\lambda$ implies that there exists $\delta > 0$ such that $\frac{d}{d\lambda}\langle K^{-1}\parl \phi \rangle ({\lambda}) < 0$ on $(\bar{\lambda} - 2 \delta, \bar{\lambda} + 2 \delta )$ such that for $\lambda \in [\bar{\lambda}, \bar{\lambda} + \delta]$, $\langle K^{-1}\parl \phi \rangle(\lambda) = \langle K^{-1}\parl \phi \rangle (\bar{\lambda}) + \int_{\bar{\lambda}}^{\lambda} d\lambda' \frac{d}{d\lambda}\langle K^{-1}\parl \phi \rangle (\lambda') \leq \langle K^{-1}\parl \phi \rangle (\bar{\lambda})  = 0$. However, this contradicts the infimum property of $\bar{\lambda}$ since there does not exist arbitrarily close $\lambda > \bar{\lambda}$ to $\bar{\lambda}$ such that $\langle K^{-1}\parl \phi \rangle (\lambda) > 0$. Thus, $\langle K^{-1}\parl \phi \rangle(\lambda) \leq 0$ for all $\lambda \geq \lambda^*$.

In this case, we automatically have $\limsup_{\lambda \to \infty} \langle K^{-1}\parl \phi \rangle \leq 0$. Next, by Eqs. \eqref{eq:redkupperpointwise} and \eqref{eq:redavgksquared}, we have the lower bound
\begin{align}\label{eq:avgrolllower}
\frac{d}{d\lambda}\langle K^{-1}\parl \phi \rangle \geq &-  \kt \left(1 + \delta e^{-\frac{2}{3}\kt^2 \lambda}\right)\left\langle U' \right\rangle  \nonumber \\ 
&- \left[ \ko^2 - \frac{3}{2}C_3 \kt^2 e^{-\left(\frac{2}{3}\kt^2 - \sqrt{2}K_{cor}^2\right)}\right]\langle K^{-1}\parl \phi \rangle \ .
\end{align}
Now, observe that $\liminf_{\lambda \to \infty} \frac{d}{d\lambda}\langle K^{-1}\parl \phi \rangle \leq 0$ since $\liminf_{\lambda \to \infty} \frac{d}{d\lambda} \langle K^{-1}\parl \phi \rangle > 0$ would result in the contradiction $\langle K^{-1}\parl \phi \rangle(\lambda) > 0$ for large enough $\lambda$. Thus, taking the $\liminf$ of \autoref{eq:avgrolllower} yields
\begin{equation}
\liminf_{\lambda \to \infty} \langle K^{-1} \parl \phi \rangle \geq - \frac{\kt }{\ko^2} \limsup_{\lambda \to \infty} \left \langle U' \right \rangle \ .
\end{equation}
The only other possible case is $\langle K^{-1}\parl \phi \rangle(\lambda) > 0$ for all $\lambda$, for which we immediately have $\liminf_{\lambda \to \infty} \langle K^{-1}\parl \phi \rangle(\lambda) \geq 0$. Now, for large enough $\lambda$, the R.H.S. of \autoref{eq:redavgksquared} is larger than $-\frac{\ko^2}{2}$. Then, \autoref{eq:avgrollevolution} becomes for large enough $\lambda$
\begin{equation}
\frac{d}{d\lambda}\langle K^{-1}\parl \phi \rangle \leq - \frac{\ko^2}{2} \langle K^{-1}\parl \phi \rangle  \ ,
\end{equation}
or in other words,
\begin{equation}
\frac{d}{d\lambda} \log \langle K^{-1}\parl \phi \rangle \leq - \frac{\ko^2}{2} \ ,
\end{equation}
which shows that $\limsup_{\lambda \to \infty} \langle K^{-1}\parl \phi \rangle \leq 0 $. In this second case, we actually have $\lim_{\lambda \to \infty} \langle K^{-1}\parl \phi \rangle  = 0 $.
\end{proof}

\begin{remark}
Consider the standard slow-roll model with a FLRW metric
\begin{equation}\label{eq:flrwmetric}
ds_{FLRW}^2 = - dt^2 + a(t)^2 d\boldsymbol{\Sigma}^2 \ ,
\end{equation}
where $d\boldsymbol{\Sigma}^2$ is a spatial metric, independent of $t$, on a three-manifold of uniform curvature---that is, the space is either elliptical (closed), euclidean (flat), or hyperbolic (open). Under the traditional slow-roll assumptions at late enough times and assuming $a(t) \to \infty$ as $t \to \infty$, we have $H^2 \coloneqq \left(\frac{\dot{a}}{a}\right)^2 \approx \frac{8\pi G}{3}U = \frac{K_U^2}{9}$ while $\dot{\phi} \approx - \frac{U'}{3H} \approx -\frac{U'}{K_U}$. Meanwhile, since $K_U \leq \kt$, we can rewrite \autoref{thm:slowrollasymptotic} as $\limsup_{\lambda \to \infty} \langle K^{-1}\parl \phi \rangle \leq 0$ and $\liminf_{\lambda \to \infty} \langle K^{-1}\parl \phi \rangle \geq - \frac{\kt^2}{\ko^2} \limsup_{\lambda \to \infty} \left \langle \frac{U'}{K_U} \right \rangle$. Hence, the asymptotic bounds in \autoref{thm:slowrollasymptotic} are reminiscent of a slow-roll inflaton in an expanding FLRW universe, up to an additional factor of $\frac{\kt^2}{\ko^2} = \frac{\lt}{\lo}$.

\end{remark}

\section{Conclusion and Future Work}\label{sec:conclusion}
In summary, we have considered 3+1 dimensional cosmologies satisfying the Einstein equations with an inflationary potential $U$ satisfying $0 < \Lambda_1 \leq U \leq \Lambda_2$ with $\frac{\Lambda_2}{\Lambda_1} < \frac{3}{2}$ and matter satisfying the dominant and the strong energy conditions. In particular, our results allow for a non-flat inflaton potential as opposed to only a cosmological constant assumed by previous work.

We have assumed that the  only potential singularities are of the crushing kind, and that the spatial slices are foliated by homogeneous but {potentially} anisotropic 2-surfaces. Though we assumed that the initial spatial slice $M_0$ has positive mean curvature $K$ everywhere, we do not require homogeneous initial conditions on $M_0$.

By probing the geometry with mean curvature flow, we show that the family of mean curvature flow surfaces $\{ M_{\lambda} \}_{\lambda \geq 0}$ has physical volume expanding in $\lambda$, at exponential rates between those of flat slicings of de Sitter spaces with cosmological constants $\Lambda_1$ and $\Lambda_2$. In particular, this result shows that inflationary expansion can indeed occur without homogeneous initial conditions and partially resolves the ``initial patch problem''. 

\textit{Future Directions: }In the limiting case of a positive cosmological constant, \cite{creminelli2020sitter} proved a de Sitter no-hair theorem where the metric converges pointwise to that of de Sitter space and the asymptotic geometry becomes physically indistinguishable from de Sitter space. Though our current bounds are too weak (in particular, the lower bound in \autoref{eq:avgksquaredlower} involves the average rather than the pointwise $K$), our partial results in \autoref{sec:auxiliary} are highly suggestive that in the particular case of a dynamical inflaton with the stress-energy tensor in \autoref{eq:inflatonstressenergy}, a similar convergence to de Sitter space may occur, given a more refined analysis of the inflaton dynamics or slightly stronger assumptions on the cosmology. Concretely, Sections \ref{sec:metricz} and \ref{sec:l1stressenergy} show that the influence of the inflaton on the inhomogeneity of the geometry is bounded on ever-expanding regions. At the same time, \autoref{sec:rolling} shows that, in some averaged sense, the inflaton is asymptotically rolling down the potential $U$ (or at least not rolling upward) so at late enough times, one expects the inflaton to converge to a constant $\Lambda_1$. Then, we might be able to recover a de Sitter no-hair theorem for a cosmological constant $\Lambda_1$.

\section*{Acknowledgments} We thank Paolo Creminelli for initial extensive collaboration in this project. We also thank Or Hershkovits and Andras Vasy for discussions.



\bibliographystyle{JHEP}

\bibliography{biblio.bib}   

\end{document}